\documentclass[11pt,journal]{IEEEtran}

\usepackage[noadjust]{cite}  % >>>
\ifCLASSINFOpdf
  \usepackage[pdftex]{graphicx}
  \graphicspath{{./figures/}}
\else
  \usepackage[dvips]{graphicx}
  \graphicspath{{./figures/}}
\fi  % >>>
\usepackage[cmex10]{amsmath}
\usepackage{amssymb}
\usepackage{amsfonts}
\usepackage{algorithmic} % >>>
\usepackage{hyperref} % >>>
\newtheorem{theorem}{Theorem}
\newtheorem{conjecture}{Conjecture}
\newtheorem{definition}{Definition}
\newtheorem{lemma}{Lemma}
\newtheorem{remark}{Remark}

\newtheorem{corollary}{Corollary}

\newcommand{\ie}{{\it i.e.}} 
\newcommand{\eg}{{\it e.g.}}  % >>>

\begin{document}
\title{Feedback-based online network coding}
\author{Jay~Kumar~Sundararajan,~\IEEEmembership{Student Member,~IEEE,}
        Devavrat~Shah,~\IEEEmembership{Member,~IEEE,}
        Muriel~M\'edard,~\IEEEmembership{Fellow,~IEEE}% <-this % stops a space
\thanks{The authors are with the Department of Electrical Engineering and Computer Science, at the Massachusetts Institute of Technology, Cambridge, MA 02139. Email: \{jaykumar, devavrat, medard\}@mit.edu}% <-this % stops a space
\thanks{Manuscript received ; revised . Parts of this work were presented at IEEE ITW 2007, IEEE ISIT 2008 and ISITA 2008.}%
}

\maketitle

\begin{abstract}
Current approaches to the practical implementation of network coding are batch-based, and often do not use feedback, except possibly to signal completion of a file download. In this paper, the various benefits of using feedback in a network coded system are studied. It is shown that network coding can be performed in a completely online manner, without the need for batches or generations, and that such online operation does not affect the throughput. Although these ideas are presented in a single-hop packet erasure broadcast setting, they naturally extend to more general lossy networks which employ network coding in the presence of feedback. The impact of feedback on queue size at the sender and decoding delay at the receivers is studied. Strategies for adaptive coding based on feedback are presented, with the goal of minimizing the queue size and delay. The asymptotic behavior of these metrics is characterized, in the limit of the traffic load approaching capacity. Different notions of decoding delay are considered, including an order-sensitive notion which assumes that packets are useful only when delivered in order. Our work may be viewed as a natural extension of Automatic Repeat reQuest (ARQ) schemes to coded networks.  
\end{abstract}

\section{Introduction}\label{intro}
% <<< Why are we writing this paper

This paper is a step towards low-delay, high-throughput solutions based on network coding, for real-time data streaming applications over a packet erasure network. In particular, it considers the role of feedback for queue management and delay control in such systems.
% >>>

\subsection{Background}
Reliable communication over a network of packet erasure channels is a well studied problem. Several solutions have been proposed, especially in the case when there is no feedback. We compare below, three such approaches -- digital fountain codes, random linear network coding and priority encoding transmission. 

%<<< Digital fountain codes
{\bf 1. Digital fountain codes: }
The digital fountain codes (\cite{ltcodes,raptor}) constitute a well-known approach to this problem. From a block of $k$ transmit packets, the sender generates random linear combinations in such a way that the receiver can, with high probability, decode the block once it receives \emph{any} set of slightly more than $k$ linear combinations.  This approach has low complexity and requires no feedback, except to signal successful decoding of the block. However, fountain codes are designed for a point-to-point erasure channel and in their original form, do not extend readily to a network setting. Consider a two-link tandem network. An end-to-end fountain code with simple forwarding at the middle node will result in throughput loss. If the middle node chooses to decode and re-encode an entire block, the scheme will be sub-optimal in terms of delay, as pointed out by \cite{linenetworks}. In this sense, the fountain code approach is not composable across links. For the special case of tree networks, there has been some recent work on composing fountain codes across links by enabling the middle node to re-encode even before decoding the entire block \cite{gummadi}.
% >>>

% <<< Random linear network coding
{\bf 2. Random linear network coding: }
Network coding was originally introduced for the case of error-free networks with specified link capacities (\cite{ahlswede,koettermedard}), and was extended to the case of erasure networks \cite{desmondthesis}. In contrast to fountain codes, the random linear network coding solution of \cite{desmondjournal} does not require decoding at intermediate nodes and can be applied in any network. Each node transmits a random linear combination of all coded packets it has received so far. This solution ensures that with high probability, the transmitted packet will have what we call the \emph{innovation guarantee property}, \ie,  it will be \emph{innovative}\footnote{An innovative packet is a linear combination of packets which is linearly independent of previously received linear combinations, and thus conveys new information.} to every receiver that receives it successfully, except if the receiver already knows as much as the sender. Thus, every successful reception will bring a unit of new information. In \cite{desmondjournal}, this scheme is shown to achieve capacity for the case of a multicast session. 
% >>>

% <<< A common problem of fountain codes and RLC
An important problem with both fountain codes and random linear network coding is that although they are rateless, the encoding operation is performed on a block (or generation) of packets. This means that in general, there is no guarantee that the receiver will be able to extract and pass on to higher layers, any of the original packets from the coded packets till the entire block has been received. This leads to a decoding delay. 

Such a decoding delay is not a problem if the higher layers will anyway use a block only as a whole (\eg, file download). This corresponds to traditional approaches in information theory where the message is assumed to be useful only as a whole. No incentive is placed on decoding ``a part of the message'' using a part of the codeword. However, many applications today involve broadcasting a continuous stream of packets in real-time (\eg, video streaming). Sources generate a stream of messages which have an intrinsic temporal ordering. In such cases, playback is possible only till the point up to which all packets have been recovered, which we call \emph{the front of contiguous knowledge}. Thus, there is incentive to decode the older messages earlier, as this will reduce the playback latency. The above schemes would segment the stream into blocks and process one block at a time. Block sizes will have to be large to ensure high throughput. However, if playback can begin only after receiving a full block, then large blocks will imply a large delay. 

This raises an interesting question: can we code in such a way that playback can begin even before the full block is received?  In other words, we are more interested in packet delay than block delay. These issues have been studied using various approaches by \cite{eminthesis}, \cite{sahai} and \cite{chris} in a point-to-point setting. However, in a network setting, the problem is not well understood. Moreover, these works do not consider the queue management aspects of the problem. In related work, \cite{sujay} and \cite{rtoblivious} address the question of how many original packets are revealed before the whole block is decoded in a fountain code setting. However, performance may depend on not only \emph{how much data} reaches the receiver in a given time, but also \emph{which part of the data}. For instance, playback delay depends on not just the number of original packets that are recovered, but also the order in which they are recovered. 
% >>>

% <<< Priority encoding transmission
{\bf 3. Priority encoding transmission: }
The scheme proposed in \cite{PET}, known as priority encoding transmission (PET), addresses this problem by proposing a code for the erasure channel that ensures that a receiver will receive the first (or highest priority) $i$ messages using the first $k_i$ coded packets, where $k_i$ increases with decreasing priority. In \cite{rankmetricPET}, \cite{walsh}, this is extended to systems that perform network coding. A concatenated network coding scheme is proposed in \cite{walsh}, with a delay-mitigating pre-coding stage. This scheme guarantees that the $k^{th}$ innovative reception will enable the receiver to decode the $k^{th}$ message. In such schemes however, the ability to decode messages in order requires a reduction in throughput because of the pre-coding stage.
% >>>

\subsection{Motivation}
The main motivation for our current work is that the availability of feedback brings the hope of simultaneously achieving the best possible throughput along with minimal packet delay and queue size. 

% <<< ARQ is great!
Reliable communication over a point-to-point packet erasure channel with full feedback can be achieved using the Automatic Repeat reQuest (ARQ) scheme -- whenever a packet gets erased, the sender retransmits it. Every successful reception conveys a new packet, implying throughput optimality. Moreover, this new packet is always the next unknown packet, which implies the lowest possible packet delay. Since there is feedback, the sender never stores anything the receiver already knows, implying optimal queue size. Thus, this simple scheme simultaneously achieves the optimal throughput along with minimal delay and queue size. Moreover, the scheme is completely online and not block-based.% >>>

% <<< But, ARQ does not work always - coding is necessary
However, if we go beyond a single point-to-point link, ARQ is not sufficient in general. Coding across packets is necessary to achieve optimal throughput, even if we allow acknowledgments. For instance, in the network coding context, link-by-link ARQ cannot achieve the multicast capacity of the \emph{butterfly network} from network coding literature \cite{ahlswede}. Similarly, ARQ is sub-optimal for broadcast-mode links because retransmitting a packet that some receivers did not get is wasteful for the others that already have it. In contrast, network coding achieves the multicast capacity of any network and also readily extends to networks with broadcast-mode links. Thus, in such situations, coding is indispensable from a throughput perspective. % >>>

% <<< Problem of interest: with coding, how to mimic ARQ's performance? 
This leads to the question -- how to combine the benefits of ARQ and network coding? The goal is to extend ARQ's desirable properties in the point-to-point context, to systems that require coding across packets. 

The problem with applying ARQ to a coded system is that a new reception may not always reveal the next unknown packet to the receiver. Instead, it may bring in a linear equation involving the packets. In conventional ARQ, upon receiving an ACK, the sender drops the ACKed packet and transmits the next one. But in a coded system, upon receiving an ACK for a linear equation, it is not clear which linear combination the sender should pick for its next transmission to obtain the best system performance. This is important because, if the receiver has to collect many equations before it can decode the unknowns involved, this could lead to a large decoding delay. 

A related question is: upon receiving the ACK for a linear equation, which packet can be excluded from future coding, \ie, which packet can be dropped from the sender's queue? If packets arrive at the sender according to some stochastic process, (as in \cite{traceyharish, desmondatilla}) and links are lossy (as in \cite{desmondjournal, desmondthesis}), then the queue management aspect of the problem also becomes important.  

One option is to drop packets that all receivers have decoded, as this would not affect the reliability. However, storing all undecoded packets may be suboptimal. Consider a situation where the sender has $n$ packets $\mathbf{p_1}, \mathbf{p_2} \ldots, \mathbf{p_n},$ and all receivers have received ($n-1$) linear combinations: ($\mathbf{p_1}$+$\mathbf{p_2}$), ($\mathbf{p_2}$+$\mathbf{p_3}$), $\ldots,$ ($\mathbf{p_{n-1}}$+$\mathbf{p_n}$). A drop-when-decoded scheme will not allow the sender to drop any packet, since no packet can be decoded by any receiver yet. However, the backlog in the amount of information, also called the \emph{virtual queue} (\cite{traceyharish, desmondatilla}), has a size of just 1. We ideally want the physical queue to track the virtual queue in size. (Indeed, in this example, it would be sufficient if the sender stores any one $\mathbf{p_i}$ in order to ensure reliable delivery.) 

These issues motivate the following questions -- if we have feedback in a system with network coding, what is the best possible tradeoff between throughput, delay and queue size? In particular, how close can we get to the performance of ARQ for the point-to-point case? These are the questions we address in this paper. % >>>

\section{Our contribution}\label{contrib}
In this paper, we show that by proper use of feedback, it is possible to perform network coding in a completely online manner similar to ARQ schemes, without the need for a block-based approach. We study the benefits of feedback in a coded network in terms of the following two aspects -- queue management and decoding delay. 

\subsection{Queue management}
\emph{Note:} In this work, we treat packets as vectors over a finite field. We restrict our attention to linear network coding. Therefore, the state of knowledge of a node can be viewed as a vector space over the field (see Section \ref{setup} for further details). 

We propose a new acknowledgment mechanism that uses feedback to \emph{acknowledge degrees of freedom\footnote{Here, \emph{degree of freedom} refers to a new dimension in the appropriate vector space representing the sender's knowledge.} instead of original decoded packets}. Based on this new form of ACKs, we propose an online coding module that naturally generalizes ARQ to coded systems. The code implies a queue update algorithm that ensures that \emph{the physical queue size at the sender will track the backlog in degrees of freedom}.

It is clear that packets that have been decoded by all receivers need not be retained at the sender. But, our proposal is more general than that. The key intuition is that we can ensure reliable transmission even if we restrict the sender's transmit packet to be chosen from a subspace that is independent\footnote{A subspace $S_1$ is said to be \emph{independent} of another subspace $S_2$ if $S_1\cap S_2 =\{\mathbf{0}\}$. See \cite{artinbook} for more details.} of the subspace representing the common knowledge available at all the receivers. 

In other words, \textit{the sender need not use for coding (and hence need not store) any information that has already been received by all the receivers}. Therefore, at any point in time, the queue simply needs to store a basis for a coset space with respect to the subspace of knowledge common to all the receivers. We define a specific way of computing this basis using the new notion of a node ``seeing'' a message packet, which is defined below.  

\begin{definition}[Index of a packet]
For any positive integer $k$, the $k^{th}$ packet that arrives at the sender is said to have an \emph{index} $k$.
\end{definition}

\begin{definition}[Seeing a packet]
A node is said to have \emph{seen} a message packet $\mathbf{p}$ if it has received enough information to compute a linear combination of the form $(\mathbf{p} + \mathbf{q})$, where $\mathbf{q}$ is itself a linear combination involving only packets with an index greater than that of $\mathbf{p}$. (Decoding implies seeing, as we can pick $\mathbf{q}=\mathbf{0}$.)
\end{definition}

In our scheme, the feedback is utilized as follows. In conventional ARQ, a receiver ACKs a packet upon decoding it successfully. However, in our scheme {\bf a receiver ACKs a packet when it sees the packet}. Our new scheme is called the \emph{drop-when-seen} algorithm because the sender \emph{drops a packet if all receivers have seen (ACKed) it}. 

Since decoding implies seeing, the sender's queue is expected to be shorter under our scheme compared to the drop-when-decoded scheme. However, we will need to show that in spite of dropping seen packets even before they are decoded, we can still ensure reliable delivery. To prove this, we present a deterministic coding scheme that uses only unseen packets and still guarantees that the coded packet will {\bf simultaneously cause each receiver that receives it successfully, to see its next unseen packet}. We will prove later that seeing a new packet translates to receiving a new degree of freedom. This means, the innovation guarantee property is satisfied and therefore, reliability and 100\% throughput can be achieved (see Algorithm 2 (b) and corresponding Theorems \ref{mainthm} and \ref{innovation} in Section \ref{algm2section}).

The intuition is that if all receivers have seen $\mathbf{p}$, then their uncertainty can be resolved using only packets with index more than that of $\mathbf{p}$ because after decoding these packets, the receivers can compute $\mathbf{q}$ and hence obtain $\mathbf{p}$ as well. Therefore, even if the receivers have not decoded $\mathbf{p}$, no information is lost by dropping it, provided it has been seen by all receivers.

Next, we present an example that explains our algorithm for a simple two-receiver case. Section \ref{formal} extends this scheme to more receivers.

\ 

{\bf Example:}
\begin{table*}
\centering
{\footnotesize
\begin{tabular}{|c|p{1in}|p{0.9in}|p{0.6in}|p{0.7in}|p{0.54in}|p{0.7in}|p{0.54in}|}
\hline
Time & Sender's queue                  & Transmitted packet	      & Channel state                     & \multicolumn{2}{c|}A         & \multicolumn{2}{c|}B          \\
     &                                 &                          &                                   & {\footnotesize Decoded} & {\footnotesize Seen but not decoded} & {\footnotesize Decoded} & {\footnotesize Seen but not decoded}  \\
\hline
1    & $\mathbf{p_1}$                           & $\mathbf{p_1}$                    & $\rightarrow$ A, $\nrightarrow$ B & $\mathbf{p_1}$                & -     & -                    & -    \\
\hline
2    & $\mathbf{p_1}$, $\mathbf{p_2}$                    & $\mathbf{p_1}\oplus \mathbf{p_2}$          & $\rightarrow$ A, $\rightarrow$  B & $\mathbf{p_1}$, $\mathbf{p_2}$         & -     & -                    & $\mathbf{p_1}$\\
\hline
3    & \ \ \ \ \ $\mathbf{p_2}$, $\mathbf{p_3}$            & $\mathbf{p_2}\oplus \mathbf{p_3}$          & $\nrightarrow$ A, $\rightarrow$  B & $\mathbf{p_1}$, $\mathbf{p_2}$         & -     & -                    & $\mathbf{p_1}, \mathbf{p_2}$\\
\hline
4    & \ \ \ \ \ \ \ \ \ \ $\mathbf{p_3}$         & $\mathbf{p_3}$                    & $\nrightarrow$ A, $\rightarrow$  B & $\mathbf{p_1}$, $\mathbf{p_2}$         & -     & $\mathbf{p_1}, \mathbf{p_2}$, $\mathbf{p_3}$    & - \\
\hline
5    & \ \ \ \ \ \ \ \ \ \ $\mathbf{p_3}$, $\mathbf{p_4}$  & $\mathbf{p_3}\oplus \mathbf{p_4}$          & $\rightarrow$ A, $\nrightarrow$ B & $\mathbf{p_1}, \mathbf{p_2}$           & $\mathbf{p_3}$ & $\mathbf{p_1}, \mathbf{p_2}, \mathbf{p_3}$      & -\\
\hline
6    & \ \ \ \ \ \ \ \ \ \ \ \ \ \ \ $\mathbf{p_4}$ & $\mathbf{p_4}$                    & $\rightarrow$ A, $\rightarrow$  B & $\mathbf{p_1}, \mathbf{p_2}, \mathbf{p_3}, \mathbf{p_4}$ & -     & $\mathbf{p_1}, \mathbf{p_2}, \mathbf{p_3}, \mathbf{p_4}$ & -    \\
\hline
\end{tabular}
}
\caption{An example of the drop-when-seen algorithm}\label{exampletable}
\vspace{-.3in}
\end{table*}
Table \ref{exampletable} shows a sample of how the proposed idea works in a packet erasure broadcast channel with two receivers A and B. The sender's queue is shown after the arrival point and before the transmission point of a slot (see Section \ref{setup} for details on the setup). In each slot, based on the ACKs, the sender identifies the next unseen packet for A and B. If they are the same packet, then that packet is sent. If not, their XOR is sent. It can be verified that with this rule, every reception causes each receiver to see its next unseen packet. 

In slot 1, $\mathbf{p_1}$ reaches A but not B. In slot 2, $(\mathbf{p_1}\oplus \mathbf{p_2})$ reaches A and B. Since A knows $\mathbf{p_1}$, it can also decode $\mathbf{p_2}$. As for B, it has now seen (but not decoded) $\mathbf{p_1}$. At this point, since A and B have seen $\mathbf{p_1}$, the sender drops it. This is fine even though B has not yet decoded $\mathbf{p_1}$, because B will eventually decode $\mathbf{p_2}$ (in slot 4), at which time it can obtain $\mathbf{p_1}$. Similarly, $\mathbf{p_2}$, $\mathbf{p_3}$ and $\mathbf{p_4}$ will be dropped in slots 3, 5 and 6 respectively. However, the drop-when-decoded policy will drop $\mathbf{p_1}$ and $\mathbf{p_2}$ in slot 4, and $\mathbf{p_3}$ and $\mathbf{p_4}$ in slot 6. Thus, our new strategy clearly keeps the queue shorter. This is formally proved in Theorem \ref{algm1qsize} and Theorem \ref{mainthm}. The example also shows that it is fine to drop packets before they are decoded. Eventually, the future packets will arrive, thereby allowing the decoding of all the packets.  

\ 

\subsubsection*{Related earlier work}
In \cite{shrader1}, Shrader and Ephremides study the queue stability and delay of block-based random linear coding versus uncoded ARQ for stochastic arrivals in a broadcast setting. However, this work does not consider the combination of coding and feedback in one scheme. In related work, \cite{shrader2} studies the case of load-dependent variable sized coding blocks with ACKs at the end of a block, using a bulk-service queue model. The main difference in our work is that receivers ACK packets even before decoding them, and this enables the sender to perform online coding.  

Sagduyu and Ephremides \cite{sagduyu1} consider online feedback-based adaptation of the code, and propose a coding scheme for the case of two receivers. This work focuses on the maximum possible stable throughput, and does not consider the use feedback to minimize queue size or decoding delay. In \cite{sagduyu2}, the authors study the throughput of a block-based coding scheme, where receivers acknowledge the successful decoding of an entire block, allowing the sender to move to the next block. Next, they consider the option of adapting the code based on feedback for the multiple receiver case. They build on the two-receiver case of \cite{sagduyu1} and propose a greedy deterministic coding scheme that may not be throughput optimal, but picks a linear combination such that the number of receivers that immediately decode a packet is maximized. In contrast, in our work we consider throughput-optimal policies that aim to minimize queue size and delay. 

In \cite{lacan}, Lacan and Lochin proposes an erasure coding algorithm called Tetrys to ensure reliability in spite of losses on the acknowledgment path. While this scheme also employs coding in the presence of feedback, their approach is to make minimal use of the feedback, in order to be robust to feedback losses. As opposed to such an approach, we investigate how best to use the available feedback to improve the coding scheme and other performance metrics. For instance, in the scheme in \cite{lacan}, packets are acknowledged (if at all) only when they are decoded, and these are then dropped from the coding window. However, we show in this work that by dropping packets when they are seen, we can maintain a smaller coding window without compromising on reliability and throughput. A smaller coding window translates to lower encoding complexity and smaller queue size at the sender in the case of stochastic arrivals.

The use of ACKs and coded retransmissions in a packet erasure broadcast channel has been considered for multiple unicasts \cite{larsson1} and multicast (\cite{jolfaei}, \cite{yongsung}, \cite{larsson3}, \cite{nguyen}). The main goal of these works however, is to optimize the throughput. Other metrics such as queue management and decoding delay are not considered. In our work, we focus on using feedback to optimize these metrics as well, in addition to achieving 100\% throughput in a multicast setting. Our coding module (in Section \ref{codingmodule}) is closely related to the one proposed by Larsson in an independent work \cite{larsson3}. However, our algorithm is specified using the more general framework of seen packets, which allows us to derive the drop-when-seen queue management algorithm and bring out the connection between the physical queue and virtual queue sizes. Reference \cite{larsson3} does not consider the queue management problem. Moreover, using the notion of seen packets allows our algorithm to be compatible even with random coding. This in turn enables a simple ACK format and makes it suitable for practical implementation. (See Remark \ref{larssonremark} for further discussion.)

\subsubsection*{Implications of our new scheme}
The newly proposed scheme has many useful implications:
\begin{itemize} 
\item {\bf Queue size:} The physical queue size is upper-bounded by the sum of the backlogs in degrees of freedom between the sender and all the receivers. This fact implies that as the traffic load approaches capacity (as load factor $\rho\rightarrow 1$), the expected size of the physical queue at the sender is $O\left(\frac1{1-\rho}\right)$. This is the same order as for single-receiver ARQ, and hence, is order-optimal.

\item {\bf Queuing analysis:} Our scheme forms a natural bridge between the virtual and physical queue sizes. It can be used to extend results on the stability of virtual queues such as \cite{traceyharish}, \cite{desmondatilla} and \cite{infocom07} to physical queues. Moreover, various results obtained for virtual queues from traditional queuing theory, such as the transform based analysis for the queue size of M/G/1 queues, or even a Jackson network type of result \cite{desmondjournal}, can be extended to the physical queue size of nodes in a network coded system.

\item {\bf Simple queue management:} Our approach based on \emph{seen packets} ensures that the sender does not have to store linear combinations of the packets in the queue to represent the basis of the coset space. Instead, it can store the basis using the original uncoded packets themselves. Therefore, the queue follows a simple first-in-first-out service discipline.

\item {\bf Online encoding:} All receivers see packets in the same order in which they arrived at the sender. This gives a guarantee that the information deficit at the receiver is restricted to a set of packets that advances in a streaming manner and has a stable size (namely, the set of unseen packets). In this sense, the proposed encoding scheme is truly online. 

\item {\bf Easy decoding:} Every transmitted linear combination is sparse -- at most $n$ packets are coded together for the $n$ receiver case. This reduces the decoding complexity as well as the overhead for embedding the coding coefficients in the packet header. 

\item {\bf Extensions:}  We present our scheme for a single packet erasure broadcast channel. However, our algorithm is composable across links and can be applied to a tandem network of broadcast links. With suitable modifications, it can potentially be applied to a more general setup like the one in \cite{desmondthesis} provided we have feedback. Such extensions are discussed further in Section \ref{extensions}.

\end{itemize}

\subsection{Decoding delay}
The drop-when-seen algorithm and the associated coding module do not guarantee that the seen packets will be decoded immediately. In general, there will be a delay in decoding, as the receiver will have to collect enough linear combinations involving the unknown packets before being able to decode the packets.

Online feedback-based adaptation of the code with the goal of minimizing decoding delay has been studied in the context of a packet erasure broadcast channel in \cite{keller}. However, their notion of delay ignores the order in which packets are decoded. For the special case of only two receivers, \cite{durvy} proposes a feedback-based coding algorithm that not only achieves 100\% throughput, but also guarantees that every successful innovative reception will cause the receiver to decode a new packet. We call this property \emph{instantaneous decodability}. However, this approach does not extend to the case of more than two receivers. With prior knowledge of the erasure pattern, \cite{keller} gives an offline algorithm that achieves optimal delay and throughput for the case of three receivers. However, in the online case, even with only three receivers, \cite{durvy} shows through an example (Example V.1) that it is not possible to simultaneously guarantee instantaneous decodability as well as throughput optimality.

In the light of this example, our current work aims for a relaxed version of instantaneous decodability while still retaining the requirement of optimal throughput. We consider a situation with stochastic arrivals and study the problem using a queuing theory approach. Let $\lambda$ and $\mu$ be the arrival rate and the channel quality parameter respectively. Let $\rho\triangleq \lambda/\mu$ be the load factor. We consider asymptotics when the load factor on the system tends to 1, while keeping either $\lambda$ or $\mu$ fixed at a number less than 1. The optimal throughput requirement means that the queue of undelivered packets is stable for all values of $\rho$ less than 1. Our new requirement on decoding delay is that the growth of the average decoding delay as a function of $\frac{1}{1-\rho}$ as $\rho\rightarrow 1$, should be of the same order as for the single receiver case. The expected per-packet delay of a receiver in a system with more than one receiver is clearly lower bounded by the corresponding quantity for a single-receiver system. Thus, instead of instantaneous decoding, we aim to guarantee asymptotically optimal decoding delay as the system load approaches capacity. The motivation is that in most practical systems, delay becomes a critical issue only when the system starts approaching its full capacity. When the load on the system is well within its capacity, the delay is usually small and hence not an issue. For the case of two receivers, it can be shown that this relaxed requirement is satisfied by the scheme in \cite{durvy} due to the instantaneous decodability property, \ie, the scheme achieves the asymptotically optimal average decoding delay per packet for the two-receiver case. 

In our current work, we provide a new coding module for the case of three receivers that {\bf achieves optimal throughput. We conjecture that at the same time it also achieves an asymptotically optimal decoding delay as the system approaches capacity, in the following sense.} With a single receiver, the optimal scheme is ARQ with no coding and we show that this achieves an expected per-packet delay at the sender of $\Theta\left(\frac{1}{1-\rho}\right)$. For the three-receiver system, we conjecture that our scheme also achieves a delay of $O\left(\frac{1}{1-\rho}\right)$, and thus meets the lower bound in an asymptotic sense. We also study a stronger notion of delay, namely the \emph{delivery delay}, which measures delay till the point when the packet can be delivered to the application above, with the constraint that packets cannot be delivered out of order. We conjecture that our scheme is asymptotically optimal even in terms of delivery delay.

We have verified these conjectures through simulations for values of $\rho$ that are very close to 1. It is useful to note that asymptotically optimal decoding delay translates to asymptotically optimal expected queue occupancy at the sender using the simple queuing rule of dropping packets that have been decoded by all receivers.

Adaptive coding allows the sender's code to incorporate receivers' states of knowledge and thereby enables the sender to control the evolution of the front of contiguous knowledge. Our schemes may thus be viewed as a step towards feedback-based control of the tradeoff between throughput and decoding delay, along the lines suggested in \cite{desmondfeedback}. 

\subsection{Organization}
The rest of the paper is organized as follows. Section \ref{setup} describes the packet erasure broadcast setting. 
Section \ref{queuesize} is concerned with adaptive codes that minimize the sender's queue size. In Section \ref{algm1section}, we define and analyze a baseline algorithm that drops packets only when they have been decoded by all receivers. Section \ref{algm2asection} presents a generic form of our newly proposed algorithm, and introduces the idea of excluding from the sender's queue, any knowledge that is common to all receivers. We show that the algorithm guarantees that the physical queue size tracks the virtual queue size. Section \ref{algm2section} presents an easily implementable variant of the generic algorithm of Section \ref{algm2asection}, called the drop-when-seen algorithm. The drop-when-seen algorithm consists of a queuing module that provides guarantees on the queue size, and a coding module that provides guarantees on reliability and throughput, while complying with the queuing module.
In Section \ref{delaysection}, we investigate adaptive codes aimed at minimizing the receivers' decoding delay. For the case of three receivers, we propose a new coding module that is proved to be throughput optimal and conjectured to be asymptotically optimal in terms of delay. Section \ref{extensions} presents some ideas on extending the algorithms to more general topologies and scenarios.
Finally, Section \ref{conc} gives the conclusions.

\section{The setup}\label{setup}
In this paper, we consider a communication problem where a sender wants to broadcast a stream of data to $n$ receivers. The data are organized into \emph{packets}, which are essentially vectors of fixed size over a finite field $\mathbb{F}_q$. A packet erasure broadcast channel connects the sender to the receivers. Time is slotted. The details of the queuing model and its dynamics are described next.

\subsection*{The queuing model}
The sender is assumed to have an infinite buffer, \ie, a queue with no preset size constraints. We assume that the sender is restricted to use linear codes. Thus, every transmission is a linear combination of packets from the incoming stream that are currently in the buffer. The vector of coefficients used in the linear combination summarizes the relation between the coded packet and the original stream. We assume that this coefficient vector is embedded in the packet header. A node can compute any linear combination whose coefficient vector is in the linear span of the coefficient vectors of previously received coded packets. In this context, the state of knowledge of a node can be defined as follows. 

\begin{definition}[Knowledge of a node]
The \emph{knowledge of a node} at some point in time is the set of all linear combinations of the original packets that the node can compute, based on the information it has received up to that point. The coefficient vectors of these linear combinations form a vector space called the \emph{knowledge space} of the node. 
\end{definition}

We use the notion of a virtual queue to represent the backlog between the sender and receiver in terms of linear degrees of freedom. This notion was also used in \cite{traceyharish}, \cite{desmondatilla} and \cite{infocom07}. There is one virtual queue for each receiver. 

\begin{definition}[Virtual queue]
For $j=1, 2, \ldots, n$, the size of the $j^{th}$ virtual queue is defined to be the difference between the dimension of the knowledge space of the sender and that of the $j^{th}$ receiver.
\end{definition}

We will use the term \emph{physical queue} to refer to the sender's actual buffer, in order to distinguish it from the virtual queues. Note that the virtual queues do not correspond to real storage. 

\begin{definition}[Degree of freedom]
  The term \emph{degree of freedom} refers to one dimension in the knowledge space of a node. It corresponds to one packet worth of data.
\end{definition}

\begin{definition}[Innovative packet]
  A coded packet with coefficient vector $\mathbf{c}$ is said to be \emph{innovative} to a receiver with knowledge space $V$ if $\mathbf{c}\notin V$. Such a packet, if successfully received, will increase the dimension of the receiver's knowledge space by one unit.
\end{definition}

\begin{definition}[Innovation guarantee property]
  Let $V$ denote the sender's knowledge space, and $V_j$ denote the knowledge space of receiver $j$ for $j=1, 2, \ldots, n$. A coding scheme is said to have the \emph{innovation guarantee property} if in every slot, the coefficient vector of the transmitted linear combination is in $V\backslash V_j$ for every $j$ such that $V_j\ne V$. In other words, the transmission is innovative to every receiver except when the receiver already knows everything that the sender knows.
\end{definition}

\subsection*{Arrivals}
Packets arrive into the sender's physical queue according to a Bernoulli process\footnote{We have assumed Bernoulli arrivals for ease of exposition. However, we expect the results to hold for more general arrival processes as well.} of rate $\lambda$. 
An arrival at the physical queue translates to an arrival at each virtual queue since the new packet is a new degree of freedom that the sender knows, but none of the receivers knows. 

\subsection*{Service}
The channel accepts one packet per slot. Each receiver either receives this packet with no errors (with probability $\mu$) or an erasure occurs (with probability $(1-\mu)$). Erasures occur independently across receivers and across slots. The receivers are assumed to be capable of detecting an erasure. 

We only consider coding schemes that satisfy the innovation guarantee property. This property implies that if the virtual queue of a receiver is not empty, then a successful reception reveals a previously unknown degree of freedom to the receiver and the virtual queue size decreases by one unit. We can thus map a successful reception by some receiver to one unit of service of the corresponding virtual queue. This means, in every slot, each virtual queue is served independently of the others with probability $\mu$.

The relation between the service of the virtual queues and the service of the physical queue depends on the queue update scheme used, and will be discussed separately under each update policy.

\subsection*{Feedback}
We assume perfect delay-free feedback. In Algorithm 1 below, feedback is used to indicate successful decoding. For all the other algorithms, the feedback is needed in every slot to indicate the occurrence of an erasure.

\subsection*{Timing}
\noindent Figure \ref{earlyarrival} shows the relative timing of various events within a slot. 
All arrivals are assumed to occur \emph{just after the beginning} of the slot. The point of transmission is after the arrival point. For simplicity, we assume very small propagation time. Specifically, we assume that the transmission, unless erased by the channel, reaches the receivers before they send feedback for that slot and feedback from all receivers reaches the sender \emph{before the end of the same slot}. Thus, the feedback incorporates the current slot's reception also. 
Based on this feedback, packets are dropped from the physical queue \emph{just before the end of the slot}, according to the queue update rule. Queue sizes are measured at the end of the slot.

\begin{figure}
	\centering
		\includegraphics[width=0.45\textwidth]{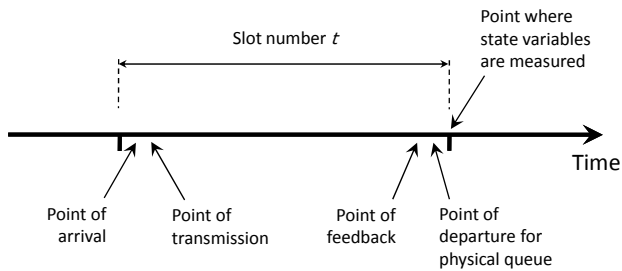}
		\caption{Relative timing of arrival, service and departure points within a slot}\label{earlyarrival}
\end{figure}

The load factor is denoted by $\rho$\ :=\ $\lambda/\mu$. In what follows, we will study the asymptotic behavior of the expected queue size and decoding delay under various policies, as $\rho\rightarrow 1$ from below. For the asymptotics, we assume that either $\lambda$ or $\mu$ is fixed, while the other varies causing $\rho$ to increase to 1.

\section{Queue size}\label{queuesize}
In this section, we first present a baseline algorithm -- retain packets in the queue until the feedback confirms that they have been decoded by all the receivers. Then, we present a new queue update rule that is motivated by a novel coding algorithm. The new rule allows the physical queue size to track the virtual queue sizes.
\subsection{Algorithm 1: Drop when decoded (baseline)}\label{algm1section}
We first present the baseline scheme which we will call Algorithm 1. It combines a random coding strategy with a drop-when-decoded rule for queue update. The coding scheme is an online version of \cite{desmondjournal} with no preset generation size -- a coded packet is formed by computing a random linear combination of all packets currently in the queue. With such a scheme, the innovation guarantee property will hold with high probability, provided the field size is large enough (We assume the field size is large enough to ignore the probability that the coded packet is not innovative. It can be incorporated into the model by assuming a slightly larger probability of erasure because a non-innovative packet is equivalent to an erasure.). 

For any receiver, the packets at the sender are unknowns, and each received linear combination is an equation in these unknowns. Decoding becomes possible whenever the number of linearly independent equations catches up with the number of unknowns involved. The difference between the number of unknowns and number of equations is essentially the backlog in degrees of freedom, \ie, the virtual queue size. Thus, \emph{a virtual queue becoming empty translates to successful decoding at the corresponding receiver}. Whenever a receiver is able to decode in this manner, it informs the sender. Based on this, the sender tracks which receivers have decoded each packet, and drops a packet if it has been decoded by all receivers. From a reliability perspective, this is fine because there is no need to involve decoded packets in the linear combination. 

\begin{remark}\label{caveat}
In general, it may be possible to solve for some of the unknowns even before the virtual queue becomes empty. For example, this could happen if a newly received linear combination cancels everything except one unknown in a previously known linear combination. It could also happen if some packets were involved in a subset of equations that can be solved among themselves locally. Then, even if the overall system has more unknowns than equations, the packets involved in the local system can be decoded. However, these are secondary effects and we ignore them in this analysis. Equivalently, we assume that if a packet is decoded before the virtual queue becomes empty, the sender ignores the occurrence of this event and waits for the next emptying of the virtual queue before dropping the packet. We believe this assumption will not change the asymptotic behavior of the queue size, since decoding before the virtual queue becoming empty is a rare event with random linear coding over a large field. 
\end{remark}

\subsubsection{The virtual queue size in steady state}
We will now study the behavior of the virtual queues in steady state. But first, we introduce some notation:

\noindent \ $Q(t)$\ \ := Size of the sender's physical queue at the end of slot $t$

\noindent \ $Q_j(t)$\ := Size of the $j^{th}$ virtual queue at the end of slot $t$

\begin{figure}
	\centering
		\includegraphics[width=0.4\textwidth]{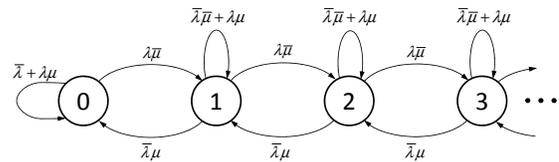}
		\caption{Markov chain representing the size of a virtual queue. Here $\bar{\lambda}:=(1-\lambda)$ and $\bar{\mu}:=(1-\mu)$. }\label{markovchain}
\end{figure}

Figure \ref{markovchain} shows the Markov chain for $Q_j(t)$. If $\lambda<\mu$, then the chain $\{Q_j(t)\}$ is positive recurrent and has a steady state distribution given by \cite{hunterbook}:
\begin{equation}
	\pi_k:=\lim_{t\rightarrow \infty} \mathbb{P}[Q_j(t)=k]=(1-\alpha)\alpha^k, \ \ \ \ \ \ \ \ k\ge 0
	\label{steadystatedist}
\end{equation}
where $\alpha=\frac{\lambda(1-\mu)}{\mu(1-\lambda)}$.

Thus, the expected size of any virtual queue in steady state is given by:
\begin{equation}
	\lim_{t\rightarrow\infty} \mathbb{E}[Q_j(t)]=\sum_{j=0}^\infty j\pi_j = (1-\mu)\cdot\frac{\rho}{(1-\rho)}
	\label{vqsize}
\end{equation}
Next, we analyze the physical queue size under this scheme.

\subsubsection{The physical queue size in steady state}
The following theorem characterizes the asymptotic behavior of the queue size under Algorithm 1, as the load on the system approaches capacity ($\rho\rightarrow 1$).

\ 

\begin{theorem}\label{algm1qsize}
\it The expected size of the physical queue in steady state for Algorithm 1 is $\Omega\left(\frac1{(1-\rho)^2}\right)$.
\end{theorem}

\ 

Comparing with Equation (\ref{vqsize}), this result makes it clear that the physical queue size does not track the virtual queue size. (We assume that $\lambda$ and $\mu$ are themselves away from 1, but only their ratio approaches 1 from below.) 

In the rest of this subsection, we present the arguments that lead to the above result. Let $T$ be the time an arbitrary arrival in steady state spends in the physical queue before departure, excluding the slot in which the arrival occurs (Thus, if a packet departs immediately after it arrives, then $T$ is 0.). A packet in the physical queue will depart when each virtual queue has become empty at least once since its arrival. Let $D_j$ be the time starting from the new arrival, till the next emptying of the $j^{th}$ virtual queue. Then, $T=\max_j D_j$ and so, $\mathbb{E}[T]\ge \mathbb{E}[D_j]$. Hence, we focus on $\mathbb{E}[D_j]$.

We condition on the event that the state seen by the new arrival just before it joins the queue, is some state $k$. There are two possibilities for the queue state at the end of the slot in which the packet arrives. If the channel is ON in that slot, then there is a departure and the state at the end of the slot is $k$. If the channel is OFF, then there is no departure and the state is $(k+1)$. Now, $D_j$ is simply the first passage time from the state at the end of that slot to state 0, \ie, the number of slots it takes for the system to reach state 0 for the first time, starting from the state at the end of the arrival slot. Let $\Gamma_{u,v}$ denote the expected first passage time from state $u$ to state $v$. The expected first passage time from state $u$ to state 0, for $u>0$ is derived in Appendix \ref{derivation}, and is given by the following expression: 

\[	\Gamma_{u,0}=\frac{u}{\mu-\lambda}\]

Now, because of the property that Bernoulli arrivals see time averages (BASTA) \cite{takagi}, an arbitrary arrival sees the same distribution for the size of the virtual queues, as the steady state distribution given in Equation (\ref{steadystatedist}). 

Using this fact, we can compute the expectation of $D_j$ as follows:
\begin{eqnarray}\label{deejay}
	\mathbb{E}[D_j]&=&\sum_{k=0}^\infty \mathbb{P}(\mbox{New arrival sees state $k$}) \mathbb{E}[D_j|\mbox{State $k$}]\nonumber\\
	&=&\sum_{k=0}^\infty \pi_k [\mu\Gamma_{k,0}+(1-\mu)\Gamma_{k+1,0}]\nonumber\\
	&=&\sum_{k=0}^\infty \pi_k \cdot \frac{\mu k+(1-\mu)(k+1)}{\mu-\lambda}\nonumber\\
	&=&\frac{1-\mu}{\mu}\cdot\frac{\rho}{(1-\rho)^2}
\end{eqnarray}

Now, the expected time that an arbitrary arrival in steady state spends in the system is given by:
\[
\mathbb{E}[T]=\mathbb{E}[\max_j D_j]\ge \mathbb{E}[D_j]=\Omega\left(\frac1{(1-\rho)^2}\right)
\]
Since each virtual queue is positive recurrent (assuming $\lambda<\mu$), the physical queue will also become empty infinitely often. Then we can use Little's law to find the expected physical queue size. 

The expected queue size of the physical queue in steady state if we use algorithm 1 is given by:
\[\lim_{t\rightarrow \infty}\mathbb{E}[Q(t)]=\lambda\mathbb{E}[T]=\Omega\left(\frac1{(1-\rho)^2}\right)\]
This discussion thus completes the proof of Theorem \ref{algm1qsize} stated above.

\subsection{Algorithm 2 (a): Drop common knowledge}\label{algm2asection}
In this section, we first present a generic algorithm that operates at the level of knowledge spaces and their bases, in order to ensure that the physical queue size tracks the virtual queue size. Later, we shall describe a simple-to-implement variant of this generic algorithm.

\subsubsection{An intuitive description}\label{strategy1}
The aim of this algorithm is to drop as much data as possible from the sender's buffer while still satisfying the reliability requirement and the innovation guarantee property. In other words, the sender should store just enough data so that it can always compute a linear combination which is simultaneously innovative to all receivers who have an information deficit. As we shall see, the innovation guarantee property is sufficient for good performance.

After each slot, every receiver informs the sender whether an erasure occurred, using perfect feedback. Thus, there is a slot-by-slot feedback requirement which means that the frequency of feedback messages is higher than in Algorithm 1. The main idea is to exclude from the queue, any knowledge that is known to all the receivers. More specifically, the queue's contents must correspond to some basis of a vector space that is independent of the intersection of the knowledge spaces of all the receivers. We show in Lemma \ref{vectorspacelemma} that with this queuing rule, it is always possible to compute a linear combination of the current contents of the queue that will guarantee innovation, as long as the field size is more than $n$, the number of receivers. 

The fact that the common knowledge is dropped suggests a modular or incremental approach to the sender's operations. Although the knowledge spaces of the receivers keep growing with time, the sender only needs to operate with the projection of these spaces on dimensions currently in the queue, since the coding module does not care about the remaining part of the knowledge spaces that is common to all receivers. Thus, the algorithm can be implemented in an incremental manner. It will be shown that this incremental approach is equivalent to the cumulative approach.  

\begin{table*}
  \centering
  \begin{tabular}{|p{1in}|p{2in}|p{2.5in}|}
	\hline
	&{\bf Uncoded Networks}&{\bf Coded Networks}\\
	\hline
	{\bf Knowledge represented by}&Set of received packets&Vector space spanned by the coefficient vectors of the received linear combinations\\
	\hline
	{\bf Amount of knowledge}&Number of packets received&Number of linearly independent (innovative) linear combinations of packets received (\ie, dimension of the knowledge space)\\
	\hline
	{\bf Queue stores}& All undelivered packets&Linear combination of packets which form a basis for the \emph{coset space} of the common knowledge at all receivers\\
	\hline
	{\bf Update rule after each transmission}&If a packet has been received by all receivers –  drop it.&Recompute the common knowledge space $V_\Delta$; Store a new set of linear combinations so that their span is independent of $V_\Delta$\\
	\hline
  \end{tabular}
  \caption{The uncoded vs. coded case}\label{intuit}
\end{table*}
Table \ref{intuit} shows the main correspondence between the notions used in the uncoded case and the coded case. We now present the queue update algorithm formally. Then we present theorems that prove that under this algorithm, the physical queue size at the sender tracks the virtual queue size. 

All operations in the algorithm occur over a finite field of size $q>n$. The basis of a node's knowledge space is stored as the rows of a basis matrix. The representation and all operations are in terms of local coefficient vectors (\emph{i.e.}, with respect to the current contents of the queue) and not global ones (\emph{i.e.}, with respect to the original packets).

\subsubsection{Formal description of the algorithm}

\noindent {\it \underline{Algorithm 2 (a)}}

\begin{enumerate}
\item [1.] Initialize basis matrices $B$, $B_1, \ldots, B_n$ to the empty matrix. These contain the bases of the incremental knowledge spaces of the sender and receivers in that order.
\item [2.] Initialize the vector $\mathbf{g}$ to the zero vector. This will hold the coefficients of the transmitted packet in each slot. \\
\noindent In every time slot, do:
\item [3.] {\it Incorporate new arrivals:}
 
Let $a$ be the number of new packets that arrived at the beginning of the slot. Place these packets at the end of the queue. Let $B$ have $b$ rows. Set $B$ to $I_{a+b}$. ($I_m$ denotes the identity matrix of size $m$.) Note that $B$ will always be an identity matrix. To make the number of columns of all matrices consistent (\emph{i.e.}, equal to $a+b$), append $a$ all-zero columns to each $B_j$. 

\item [4.] {\it Transmission: }

  If $B$ is not empty, update $\mathbf{g}$ to be any vector that is in $span(B)$, but not in $\cup_{\{j: B_j\subsetneq B\}} span(B_j)$. (Note: $span(B)$ denotes the row space of $B$.)
\ 

Lemma~\ref{vectorspacelemma} shows that such a $\mathbf{g}$ exists. Let $\mathbf{y_1, y_2, }\ldots \mathbf{y_Q}$ represent the current contents of the queue, where the queue size $Q=(a+b)$. Compute the linear combination $\sum_{i=1}^{Q} g_i \mathbf{y_i}$ and transmit it on the packet erasure broadcast channel. If $B$ is empty, set $\mathbf{g}$ to $\mathbf{0}$ and transmit nothing.

\item [5.] {\it Incorporate feedback: }

Once the feedback arrives, for every receiver $j=1$ to $n$, do:
\begin{enumerate}
\item[] If $\mathbf{g}\ne \mathbf{0}$ and the transmission was successfully received by receiver $j$ in this slot, append $\mathbf{g}$ as a new row to $B_j$.
\end{enumerate}

\item [6.] {\it Separate out the knowledge that is common to all receivers: } 

Compute the following (the set notation used here considers the matrices as a set of row vectors):

\begin{tabular}{lp{.1in}p{2.3in}}
$B_{\Delta}$&:=& Any basis of $\cap_{j=1}^n span(B_j)$.\\
$B'$&:=& Completion of $B_{\Delta}$ into a basis of $span(B)$.\\
$B''$&:=& $B'\backslash B_{\Delta}$.\\
$B_{j}'$&:=& Completion of $B_{\Delta}$ into a basis of $span(B_j)$ in such a way that, if we define $B_j'':=B_j'\backslash B_{\Delta}$, then the following holds: $B_j'' \subseteq span(B'')$. Lemma~\ref{bjlemma} proves that this is possible. \\
\end{tabular}

\item [7.] {\it Update the queue contents: }

Replace the contents of the queue with packets $\mathbf{y_1', y_2', }\ldots \mathbf{y_{Q'}'}$ of the form $\sum_{i=1}^Q h_i \mathbf{y_i}$ for each $\mathbf{h}\in B''$. The new queue size $Q'$ is thus equal to the number of rows in $B''$.

\item [8.] {\it Recompute local coefficient vectors with respect to the new queue contents: }

Find a matrix $C_j$ such that $B_j''=X_jB''$ (this is possible because $B_j'' \subseteq span(B'')$). Call $X_j$ the new $B_j$. 
  Update the value of $B$ to $I_{Q'}$.

\item [9.] Go back to step 3 for the next slot.
\end{enumerate}

The above algorithm essentially removes, at the end of each slot, the common knowledge (represented by the basis $B_{\Delta}$) and retains only the remainder $B''$. The knowledge spaces of the receivers are also represented in an incremental manner in the form of $B_j''$, excluding the common knowledge. Since $B_j''\subseteq span(B'')$, the $B_j''$ vectors can be completely described in terms of the vectors in $B''$. It is as if $B_{\Delta}$ has been completely removed from the entire setting, and the only goal remaining is to convey $span(B'')$ to the receivers. Hence, it is sufficient to store linear combinations corresponding to $B''$ in the queue. $B''$ and $B_j''$ get mapped to the new $B$ and $B_j$, and the process repeats in the next slot. 

\begin{lemma}\label{bjlemma}
\it In step 5 of the algorithm above, it is possible to complete $B_{\Delta}$ into a basis $B_j'$ of each $span(B_j)$ such that $B_j''\subseteq span(B'')$.
\end{lemma}

\begin{proof}
We show that any completion of $B_{\Delta}$ into a basis of $span(B_j)$ can be changed to a basis with the required property. 

Let $B_{\Delta}=\{\mathbf{b_1}, \mathbf{b_2}, \ldots , \mathbf{b_m}\}$. Suppose we complete this into a basis $C_j$ of $span(B_j)$ such that:
\[C_j=B_{\Delta}\cup \{\mathbf{c_1}, \mathbf{c_2}, \ldots , \mathbf{c_{|B_j|-m}}\}\]

Now, we claim that at the beginning of step 6, $span(B_j)\subseteq span(B)$ for all $j$. This can be proved by induction on the slot number, using the way the algorithm updates $B$ and the $B_j$'s.  Intuitively, it says that any receiver knows a subset of what the sender knows.

Therefore, for each vector $\mathbf{c} \in C_j\backslash B_{\Delta}$, $\mathbf{c}$ must also be in $span(B)$. Now, since $B_{\Delta}\cup B''$ is a basis of $span(B)$, we can write $\mathbf{c}$ as $\sum_{i=1}^m \alpha_i \mathbf{b_i} + \mathbf{c'}$ with $\mathbf{c'}\in span(B'')$. In this manner, each $\mathbf{c_i}$ gives a distinct $\mathbf{c_i'}$. It is easily seen that $C_j':=B_{\Delta}\cup \{\mathbf{c_1'}, \mathbf{c_2'}, \ldots , \mathbf{c_{|B_j|-m}'}\}$ is also a basis of the same space that is spanned by $C_j$. Moreover, it satisfies the property that $C_j'\backslash B_{\Delta} \subseteq span(B'')$.
\end{proof}

\begin{lemma}\label{vectorspacelemma}
(\cite{infocom07}) \it \ Let $\mathcal{V}$ be a vector space with dimension $k$ over a field of size $q$, and let $\mathcal{V}_1, \mathcal{V}_2, \ldots \mathcal{V}_n$, be subspaces of $\mathcal{V}$, of dimensions $k_1, k_2, \ldots, k_n$ respectively. Suppose that $k>k_i$ for all $i=1, 2, \ldots, n$. Then, there exists a vector that is in $\mathcal{V}$ but is not in any of the $\mathcal{V}_i$'s, if $q>n$.
\end{lemma}

\begin{proof} See \cite{infocom07} for the proof.
\end{proof}
This lemma is also closely related to the result in \cite{larsson3}, which derives the smallest field size needed to ensure innovation guarantee. 

\subsubsection{Connecting the physical and virtual queue sizes}\label{proofsection}
In this subsection, we will prove the following result that relates the size of the physical queue at the sender and the virtual queues, which themselves correspond to the backlog in degrees of freedom.

\ 

\begin{theorem}\label{algm2athm}
\it For Algorithm 2 (a), the physical queue size at the sender is upper bounded by the sum of the backlog differences between the sender and each receiver in terms of the number of degrees of freedom.
\end{theorem}

\ 

Let $a(t)$ denote the number of arrivals in slot $t$, and let $A(t)$ be the total number of arrivals up to and including slot $t$, \ie, $A(t)=\sum_{t'=0}^t a(t')$. Let $B(t)$ (resp. $B_j(t)$) be the matrix $B$ (resp. $B_j$) after incorporating the slot $t$ arrivals, \emph{i.e.}, at the end of step 3 in slot $t$. Let $H(t)$ be a matrix whose rows are the \emph{global} coefficient vectors of the queue contents at the end of step 3 in time slot $t$, \ie, the coefficient vectors in terms of the original packet stream. Note that each row of $H(t)$ is in $\mathbb{F}_q^{A(t)}$. 

Let $\mathbf{g}(t)$ denote the vector $\mathbf{g}$ at the calculated in step 4 in time slot $t$, \emph{i.e.}, the local coefficient vector of the packet transmitted in slot $t$. Also, let $B_\Delta(t)$ (resp. $B''(t)$, $B_j'(t)$ and $B_j''(t)$) denote the matrix  $B_\Delta$ (resp. $B''$, $B_j'$ and $B_j''$) at the end of step 6 in time slot $t$. 

\begin{lemma}\label{localglobalmaplemma}
\it The rows of $H(t)$ are linearly independent for all $t$.
\end{lemma}

\begin{proof}
The proof is by induction on $t$. 

{\it Basis step:} In the beginning of time slot 1, $a(1)$ packets arrive. So, $H(1)=I_{a(1)}$ and hence the rows are linearly independent. 

{\it Induction hypothesis:} Assume $H(t-1)$ has linearly independent rows.

{\it Induction step:} The queue is updated such that the linear combinations corresponding to local coefficient vectors in $B''$ are stored, and subsequently, the $a(t)$ new arrivals are appended. Thus, the relation between $H(t-1)$ and $H(t)$ is:
\[H(t)=\left[ \begin{array}{cc} B''(t-1)H(t-1) & 0 \\ 0 & I_{a(t)} \end{array} \right]\]

Now, $B''(t-1)$ has linearly independent rows, since the rows form a basis. The rows of $H(t-1)$ are also linearly independent by hypothesis. Hence, the rows of $B''(t-1)H(t-1)$ will also be linearly independent. Appending $a(t)$ zeros and then adding an identity matrix block in the right bottom corner does not affect the linear independence. Hence, $H(t)$ also has linearly independent rows.
\end{proof}

Define the following:

\ 

\begin{tabular}{lcp{2.5in}}
$U(t)$&:=&Row span of $H(t)$\\
$U_j(t)$&:=&Row span of $B_j(t)H(t)$\\
$U_j'(t)$&:=&Row span of $B_j'(t)H(t)$\\
$U_\Delta '(t)$ &:=& $\cap_{j=1}^n U_j'(t)$\\
$U''(t)$&:=&Row span of $B''(t)H(t)$\\
$U_j''(t)$&:=&Row span of $B_j''(t)H(t)$\\
\end{tabular}
\ 

All the vector spaces defined above are subspaces of $\mathbb{F}_q^{A(t)}$. Figure \ref{timeline} shows the points at which these subspaces are defined in the slot.

\begin{figure}[t]
\centering
\includegraphics[width=0.45\textwidth]{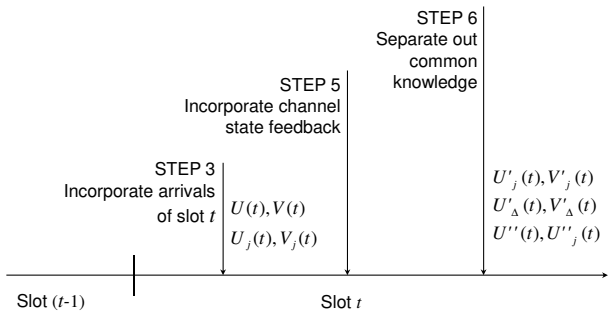}
\caption{The main steps of the algorithm, along with the times at which the various $U(t)$'s are defined}\label{timeline}
\vspace{-.2in}
\end{figure}

The fact that $H(t)$ has full row rank (proved above in Lemma~\ref{localglobalmaplemma}) implies that the operations performed by the algorithm in the domain of the local coefficient vectors can be mapped to the corresponding operations in the domain of the global coefficient vectors:

\begin{enumerate}
\item The intersection subspace $U_\Delta '(t)$ is indeed the row span of $B_\Delta(t)H(t)$. 

\item Let $R_j(t)$ be an indicator (0-1) random variable which takes the value 1 iff the transmission in slot $t$ is successfully received without erasure by receiver $j$ and in addition, receiver $j$ does not have all the information that the sender has. Let $\mathbf{\tilde{g_j}}(t):=R_j(t)\mathbf{g}(t)H(t)$. Then,
\begin{equation}\label{tildeeqn}
U_j'(t)=U_j(t)\oplus span(\mathbf{\tilde{g_j}}(t))
\end{equation}
where $\oplus$ denotes direct sum of vector spaces. The way the algorithm chooses $\mathbf{g}(t)$ guarantees that if $R_j(t)$ is non-zero, then $\mathbf{\tilde{g_j}}(t)$ will be outside the corresponding $U_j(t)$, \emph{i.e.}, it will be innovative. This fact is emphasized by the direct sum in this equation.

\item Because of the way the algorithm performs the completion of the bases in the local domain in step 6, the following properties hold in the global domain:
\begin{eqnarray}
U(t)&=&U_\Delta '(t)\oplus U''(t)\label{eqn1}\\
U_j'(t)&=&U_\Delta '(t)\oplus U_j''(t) \mbox{ \ \ \ \ and, }\label{eqn2}\\
U_j''(t)&\subseteq &U''(t),\ \ \ \  \forall j=1, 2, \ldots, n\label{eqn3}
\end{eqnarray}
\end{enumerate}

From the above properties, we can infer that $U_1''(t)+U_2''(t)+\ldots U_n''(t)\subseteq U''(t)$. After incorporating the arrivals in slot $t+1$, this gives $U_1(t+1)+U_2(t+1)+\ldots U_n(t+1)\subseteq U(t+1)$. Since this is true for all $t$, we write it as:
\begin{equation}\label{equat}
U_1(t)+U_2(t)+\ldots U_n(t)\subseteq U(t)
\end{equation}

Now, in order to relate the queue size to the backlog in number of degrees of freedom, we define the following vector spaces which represent the \emph{cumulative} knowledge of the sender and receivers (See Figure \ref{timeline} for the timing):

\ 

\begin{tabular}{ccp{2.5in}}
 $V(t)$ &:=&  Sender's knowledge space after incorporating the arrivals (at the end of step 3) in slot $t$. This is simply equal to $\mathbb{F}_q^{A(t)}$\\
 $V_j(t)$ &:=&  Receiver $j$'s knowledge space at the end of step 3 in slot $t$\\
 $V_j'(t)$ &:=& Receiver $j$'s knowledge space in slot $t$, after incorporating the channel state feedback into $V_j(t)$, \emph{i.e.}, $V_j'(t)=V_j(t)\oplus span(\mathbf{\tilde{g_j}}(t))$.\\
 $V_\Delta(t)$ &:=& $\cap_{j=1}^n V_j(t)$\\
 $V_\Delta '(t)$ &:=& $\cap_{j=1}^n V_j'(t)$\\
\end{tabular}

\ 

For completeness, we now prove the following facts about direct sums of vector spaces that we will use.

\begin{lemma}\label{distrib_lemma}
\it Let $V$ be a vector space and let $V_\Delta, U_1, U_2, \ldots U_n$ be subspaces of $V$ such that, $V_\Delta$ is independent of the span of all the $U_j$'s, \emph{i.e.,} $dim [V_\Delta\cap (U_1+U_2+\ldots +U_n)]=0$. Then, \[V_\Delta \oplus \left[\cap_{i=1}^n U_i\right]=\cap_{i=1}^n \left[V_\Delta \oplus U_i\right] \]
\end{lemma}

See Appendix \ref{distriblemmaproof} for the proof.

\ 

\begin{lemma}\label{associativity}
  \it
  Let $A, B,$ and $C$ be three vector spaces such that $B$ is independent of $C$ and $A$ is independent of $B\oplus C$. Then the following hold:
  \begin{enumerate}
	\item $A$ is independent of $B$.\label{statement1}
	\item $A\oplus B$ is independent of $C$.\label{statement2}
	\item $A\oplus (B\oplus C) = (A\oplus B)\oplus C$.\label{statement3}
  \end{enumerate}
\end{lemma}

See Appendix \ref{associativityproof} for the proof.

\

\begin{theorem}\label{incrementtheorem}
\it For all $t\ge 0$,
\begin{eqnarray*}
V(t)&=& V_\Delta (t)\oplus U(t)\\
V_j(t)&=& V_\Delta (t)\oplus U_j(t) \ \ \ \ \forall j=1,2,\ldots n\\
V_\Delta'(t)&=& V_\Delta(t)\oplus U_\Delta'(t)
\end{eqnarray*}
\end{theorem}

\begin{proof}
The proof is by induction on $t$. 

\noindent {\it Basis step:}

At $t=0$, $V(0)$, $U(0)$ as well as all the $V_j(0)$'s and $U_j(0)$'s are initialized to $\{\mathbf{0}\}$. Consequently, $V_\Delta(0)$ is also $\{\mathbf{0}\}$. It is easily seen that these initial values satisfy the equations in the theorem statement.

\ 

\noindent {\it Induction Hypothesis:}

We assume the equations hold at $t$, \emph{i.e.}, 
\begin{eqnarray}
  V(t)&\mbox{=}& V_\Delta (t)\oplus U(t)\label{ih1}\\
  V_j(t)&\mbox{=}& V_\Delta (t)\oplus U_j(t), \forall j=1, 2, \ldots n\label{ih2}\\
  V_\Delta'(t)&\mbox{=}& V_\Delta(t)\oplus U_\Delta'(t)\label{ih3}
\end{eqnarray}

\noindent {\it Induction Step:} We now prove that they hold in slot $(t+1)$. We have:

$V(t)$
\begin{align*}
  &= V_\Delta(t) \oplus U(t) &&\text{(from (\ref{ih1}))}\\
  &= V_\Delta (t)\oplus [U_\Delta '(t)\oplus U''(t)]&& \text{(from (\ref{eqn1}))}\\
  &= [V_\Delta (t)\oplus U_\Delta '(t)]\oplus U''(t)&& \text{(Lemma \ref{associativity})}\\
  &= V_\Delta' (t) \oplus U''(t) && \text{(from (\ref{ih3}))}
\end{align*}
Thus, we have proved:
\begin{equation}\label{firststep}
  V(t)=V_\Delta' (t) \oplus U''(t)
\end{equation}

Now, we incorporate the arrivals in slot $(t+1)$. This converts $V_\Delta'(t)$ to $V_\Delta(t+1)$, $U''(t)$ to $U(t+1)$, and $V(t)$ to $V(t+1)$, due to the following operations:
\begin{eqnarray*}
\mbox{Basis of } V_\Delta(t+1)&=&\left[\begin{array}{cc}\mbox{Basis of } V_\Delta'(t) & 0\end{array}\right]\\
\mbox{Basis of } U(t+1)\ \ &=&\left[\begin{array}{cc}\mbox{Basis of } U''(t) & 0 \\ 0 & I_{a(t+1)}\end{array}\right]\\
\mbox{Basis of } V(t+1)\ \ &=&\left[\begin{array}{cc}\mbox{Basis of } V(t) & 0 \\ 0 & I_{a(t+1)}\end{array}\right]
\end{eqnarray*}

Incorporating these modifications into (\ref{firststep}), we get: 
\[V(t+1)=V_\Delta(t+1)\oplus U(t+1)\]

Now, consider each receiver $j=1, 2, \ldots n$.

$V_j'(t)$
\begin{align*}
&= V_j(t) \oplus span(\mathbf{\tilde{g_j}}(t))&&\\
&= [V_\Delta(t) \oplus U_j(t)] \oplus span(\mathbf{\tilde{g_j}}(t))&& \text{(from (\ref{ih2}))}\\
&= V_\Delta(t) \oplus [U_j(t) \oplus span(\mathbf{\tilde{g_j}}(t))]&& \text{(Lemma \ref{associativity})}\\
&= V_\Delta (t) \oplus U_j'(t)&& \text{(from (\ref{tildeeqn}))}\\
&= V_\Delta (t)\oplus [U_\Delta '(t)\oplus U_j''(t)] && \text{(from (\ref{eqn2}))}\\
&= [V_\Delta (t)\oplus U_\Delta '(t)]\oplus U_j''(t) && \text{(Lemma \ref{associativity})}\\
&= V_\Delta' (t) \oplus U_j''(t) && \text{(from (\ref{ih3}))}
\end{align*}

Incorporating the new arrivals into the subspaces involves adding $a(t+1)$ all-zero columns to the bases of $V_j'(t)$, $V_\Delta'(t)$, and $U_j''(t)$, thereby converting them into bases of $V_j(t+1), V_\Delta(t+1)$, and $U_j(t+1)$ respectively. These changes do not affect the above relation, and we get:
\[V_j(t+1)=V_\Delta(t+1)\oplus U_j(t+1),\ \ \ \ \forall j=1, 2, \ldots n\]
And finally,

$V_\Delta' (t+1) $
\begin{eqnarray*}
&=& \cap_{j=1}^n V_j'(t+1)\\
&=& \cap_{j=1}^n [V_j(t+1)\oplus span(\mathbf{\tilde{g_j}}(t+1))]\\
&=& \cap_{j=1}^n [V_\Delta (t+1)\oplus U_j(t+1)\oplus span(\mathbf{\tilde{g_j}}(t+1))]\\
&\stackrel{(a)}{=}& V_\Delta (t+1) \oplus \cap_{j=1}^n [U_j(t+1) \oplus span(\mathbf{\tilde{g_j}}(t+1))]\\
&=& V_\Delta (t+1)\oplus U_\Delta'(t+1)\\
\end{eqnarray*}
Step $(a)$ is justified as follows. 
Using equation (\ref{equat}) and the fact that $\mathbf{\tilde{g_j}}(t+1)$ was chosen to be inside $U(t+1)$, we can show that the span of all the $[U_j(t+1)\oplus span(\mathbf{\tilde{g_j}}(t+1))]$'s is inside $U(t+1)$. Now, from the induction step above, $V_\Delta(t+1)$ is independent of $U(t+1)$. Therefore, $V_\Delta(t+1)$ is independent of the span of all the $[U_j(t+1)\oplus span(\mathbf{\tilde{g_j}}(t+1))]$'s. We can therefore apply Lemma \ref{distrib_lemma}. 
\end{proof}

\begin{theorem}\label{thm1}
\it Let $Q(t)$ denote the size of the queue after the arrivals in slot $t$ have been appended to the queue.  
\[Q(t)=dim\ V(t)-dim\ V_\Delta(t)\]
\end{theorem}

\begin{proof}

\noindent	$Q(t) = dim\ U(t) = dim\ U''(t-1)+a(t)$
\begin{align*}
	&= dim\ U(t-1)-dim\ U_\Delta '(t-1)+a(t) \\
	& \ \ \ \ \ \ \mbox{(using (\ref{eqn1})}\\ 
	&= dim\ V(t-1) - dim\ V_\Delta (t-1) -dim\ U_\Delta'(t) +a(t) \\
	& \ \ \ \ \ \ \mbox{(from Theorem \ref{incrementtheorem})}\\
	&= dim\ V(t-1) - dim\ V_\Delta'(t)+a(t)\\
	& \ \ \ \ \ \ \mbox{ (from Theorem \ref{incrementtheorem})}\\
	&= dim\ V(t)-dim\ V_\Delta (t)
\end{align*}
\end{proof}

\begin{lemma}\label{modularitylemma1}
\it Let $V_1, V_2, \ldots, V_k$ be subspaces of a vector space $V$. Then, for $k\ge 1$,
\[dim(V_1\cap V_2\cap \ldots \cap V_k)\ge \sum_{i=1}^k dim(V_i) - (k-1)dim(V)\]
\end{lemma}

\begin{proof}
For any two subspaces $X$ and $Y$ of $V$, 
\[dim(X\cap Y)+dim(X+Y)=dim(X)+dim(Y) \]
where $X+Y$ denotes the span of subspaces $X$ and $Y$.

Hence,
\begin{eqnarray}\label{dimineq}
\nonumber dim(X\cap Y) &=& dim(X)+dim(Y)-dim(X+Y)\\
&\ge & dim(X)+dim(Y)-dim(V)\\
\nonumber &&\mbox{(since $X+Y$ is also a subspace of $V$)}
\end{eqnarray}

Now, we prove the lemma by induction on $k$.

\noindent {\it Basis step:}

$k=1:$ LHS = $dim(V_1)$, \ RHS = $dim(V_1)$

$k=2:$ LHS = $dim(V_1\cap V_2)$, \ RHS = $dim(V_1)+dim(V_2)-dim(V)$\\
The claim follows from inequality (\ref{dimineq}).

\noindent {\it Induction Hypothesis:}

For some arbitrary $k$, 
\[dim(\cap_{i=1}^{k-1} V_i)\ge \sum_{i=1}^{k-1} dim(V_i)-(k-2)dim(V) \]

\noindent {\it Induction Step:}

$dim(\cap_{i=1}^k V_i) = dim(V_k\cap \cap_{i=1}^{k-1} V_i)$
\begin{eqnarray*}
&\ge & dim(V_k)+dim(\cap_{i=1}^{k-1} V_i)-dim(V) \ \ \ \mbox{(using (\ref{dimineq}))}\\
&\ge & dim(V_k)+\left[\sum_{i=1}^{k-1} dim(V_i) -(k-2)dim(V)\right] \\
&&-dim(V)\\
&=& \sum_{i=1}^{k} dim(V_i) -(k-1)dim(V)
\end{eqnarray*}

The above result can be rewritten as:
\begin{equation}\label{connectingineq1}
dim(V)-dim(V_1\cap V_2\cap \ldots V_k)\le \sum_{i=1}^k[dim(V)-dim(V_i)]
\end{equation}
\end{proof}

Using this result, we can now prove Theorem \ref{algm2athm}.

{\it Proof of Theorem \ref{algm2athm}:}
If we apply Lemma~\ref{modularitylemma1} to the vector spaces $V_j(t), j=1,2,\ldots ,n$ and $V(t)$, then the left hand side of inequality (\ref{connectingineq1}) becomes the sender queue size (using Theorem~\ref{thm1}), while the right hand side becomes the sum of the differences in backlog between the sender and the receivers, in terms of the number of degrees of freedom. Thus, we have proved Theorem \ref{algm2athm}.
\endproof

\subsection{Algorithm 2 (b): Drop when seen}\label{algm2section}
The drop-when-seen algorithm can be viewed as a specialized variant of the generic Algorithm 2 (a) given above. It uses the notion of seen packets (defined in Section \ref{contrib}) to represent the bases of the knowledge spaces. This leads to a simple and easy-to-implement version of the algorithm which, besides ensuring that physical queue size tracks virtual queue size, also provides some practical benefits. For instance, the sender need not store linear combinations of packets in the queue like in Algorithm 2 (a). Instead only original packets need to be stored, and the queue can be operated in a simple first-in-first-out manner. We now present some mathematical preliminaries before describing the algorithm.

\subsubsection{Some preliminaries}
The newly proposed algorithm uses the notion of reduced row echelon form (RREF) of a matrix to represent the knowledge of a receiver. Hence, we first recapitulate the definition and some properties of the RREF from \cite{artinbook}, and present the connection between the RREF and the notion of seeing packets.

\begin{definition}[Reduced row echelon form (RREF)]
	A matrix is said to be in reduced row echelon form if it satisfies the following conditions:
	\begin{enumerate}
		\item The first nonzero entry of every row is 1. 
		\item The first nonzero entry of any row is to the right of the first nonzero entry of the previous row.
		\item The entries above the first nonzero row of any row are all zero.
	\end{enumerate}
\end{definition}

The RREF leads to a standard way to represent a vector space. Given a vector space, consider the following operation -- arrange the basis vectors in any basis of the space as the rows of a matrix, and perform Gaussian elimination. This process essentially involves a sequence of elementary row transformations and it produces a unique matrix in RREF such that its row space is the given vector space. We call this the RREF basis matrix of the space. We will use this representation for the knowledge space of the receivers.

Let $V$ be the knowledge space of some receiver. Suppose $m$ packets have arrived at the sender so far. 
Then the receiver's knowledge consists of linear combinations of some collection of these $m$ packets, \emph{i.e.}, $V$ is a subspace of $\mathbb{F}_q^m$. Using the procedure outlined above, we can compute the $dim(V)\times m$ RREF basis matrix of $V$ over $\mathbb{F}_q$. 

In the RREF basis, the first nonzero entry of any row is called a \emph{pivot}. Any column with a pivot is called a \emph{pivot column}. By definition, each pivot occurs in a different column. Hence, the number of pivot columns equals the number of nonzero rows, which is $dim[V]$.
Let $\mathbf{p_k}$ denote the packet with index $k$. The columns are ordered so that column $k$ maps to packet $\mathbf{p_k}$. The following theorem connects the notion of seeing packets to the RREF basis. 

\begin{theorem}\label{rreftheorem}
	\it A node has seen a packet with index $k$ if and only if the $k^{th}$ column of the RREF basis $B$ of the knowledge space $V$ of the node is a pivot column. 
\end{theorem}

\IEEEproof
The `if' part is clear. If column $k$ of $B$ is a pivot column, then the corresponding pivot row corresponds to a linear combination known to the node, of the form $\mathbf{p_k}+\mathbf{q}$, where $\mathbf{q}$ involves only packets with index more than $k$. Thus, the node has seen $\mathbf{p_k}$.

For the `only if' part, suppose column $k$ of $B$ does not contain a pivot. Then, in any linear combination of the rows, rows with pivot after column $k$ cannot contribute anything to column $k$. Rows with pivot before column $k$ will result in a non-zero term in some column to the left of $k$. Since every vector in $V$ is a linear combination of the rows of $B$, the first non-zero term of any vector in $V$ cannot be in column $k$. Thus, $\mathbf{p_k}$ could not have been seen. 
\endproof

Since the number of pivot columns is equal to the dimension of the vector space, we obtain the following corollary.
\begin{corollary}\label{numberseen}
	\it The number of packets seen by a receiver is equal to the dimension of its knowledge space.
\end{corollary}
The next corollary introduces a useful concept. 
\begin{corollary}\label{witness}
	\it If receiver $j$ has seen packet $\mathbf{p_k}$, then it knows exactly one linear combination of the form $\mathbf{p_k}+\mathbf{q}$ such that $\mathbf{q}$ involves only \emph{\textbf{unseen}} packets with index more than $k$. 
\end{corollary}
\IEEEproof
We use the same notation as above. The receiver has seen $\mathbf{p_k}$. Hence, column $k$ in $B$ is a pivot column. By definition of RREF, in the row containing the pivot in column $k$, the pivot value is 1 and subsequent nonzero terms occur only in non-pivot columns. Thus, the corresponding linear combination has the given form $\mathbf{p_k}+\mathbf{q}$, where $\mathbf{q}$ involves only \emph{unseen} packets with index more than $k$. 

We now prove uniqueness by contradiction. Suppose the receiver knows another such linear combination $\mathbf{p_k}+\mathbf{q'}$ where $\mathbf{q'}$ also involves only unseen packets. Then, the receiver must also know $(\mathbf{q}-\mathbf{q'})$. But this means the receiver has seen some packet involved in either $\mathbf{q}$ or $\mathbf{q'}$ -- a contradiction. 
\endproof

\begin{definition}[Witness]
	We denote the unique linear combination guaranteed by Corollary \ref{witness} as $\mathbf{W_j}(\mathbf{p_k})$, the \emph{witness for receiver $j$ seeing $\mathbf{p_k}$}.
\end{definition}

\subsubsection{Description of Algorithm 2 (b)}\label{strategy2}
The central idea of the algorithm is to keep track of seen packets instead of decoded packets. The two main parts of the algorithm are the coding and queue update modules.  

In Section \ref{codingmodule}, we present the formal description of our coding module. The coding module computes a linear combination $\mathbf{g}$ that will cause any receiver that receives it, to see its next unseen packet. First, for each receiver, the sender computes its knowledge space using the feedback and picks out its next unseen packet. Only these packets will be involved in $\mathbf{g}$, and hence we call them the \emph{transmit set}. Now, we need to select coefficients for each packet in this set. Clearly, the receiver(s) waiting to see the oldest packet in the transmit set (say $\mathbf{p_1}$) will be able to see it as long as its coefficient is not zero. Consider a receiver that is waiting to see the second oldest packet in the transmit set (say $\mathbf{p_2})$. Since the receiver has already seen $\mathbf{p_1}$, it can subtract the witness for $\mathbf{p_1}$, thereby canceling it from $\mathbf{g}$. The coefficient of $\mathbf{p_2}$ must be picked such that after subtracting the witness for $\mathbf{p_1}$, the remaining coefficient of $\mathbf{p_2}$ in $\mathbf{g}$ is non-zero. The same idea extends to the other coefficients. The receiver can cancel packets involved in $\mathbf{g}$ that it has already seen by subtracting suitable multiples of the corresponding witnesses. Therefore, the coefficients for $\mathbf{g}$ should be picked such that for each receiver, after canceling the seen packets, the remaining coefficient of the next unseen packet is non-zero. Then, the receiver will be able to see its next unseen packet. Theorem \ref{innovation} proves that this is possible if the field size is at least $n$, the number of receivers. With two receivers, the coding module is a simple XOR based scheme (see Table \ref{exampletable}). Our coding scheme meets the innovation guarantee requirement because Theorem \ref{rreftheorem} implies that a linear combination that would cause a new packet to be seen brings in a previously unknown degree of freedom.

The fact that the coding module uses only the next unseen packet of all receivers readily implies the following queue update rule. \textbf{Drop a packet if all receivers have seen it.} This simple rule ensures that the physical queue size tracks the virtual queue size. 

\begin{remark}\label{larssonremark}
  In independent work, \cite{larsson3} proposes a coding algorithm which uses the idea of selecting those packets for coding, whose indices are one more than each receiver's rank. This corresponds to choosing the next unseen packets in the special case where packets are seen in order. Moreover, this algorithm picks coding coefficients in a deterministic manner, just like our coding module. Therefore, our module is closely related to the algorithm of \cite{larsson3}. 
  
  However, our algorithm is based on the framework of seen packets. This allows several benefits. First, it immediately leads to the drop-when-seen queue management algorithm, as described above. In contrast, \cite{larsson3} does not consider queuing aspects of the problem. Second, in this form, our algorithm readily generalizes to the case where the coding coefficients are picked randomly. The issue with random coding is that packets may be seen out of order. Our algorithm will guarantee innovation even in this case (provided the field is large), by selecting a random linear combination of the next unseen packets of the receivers. However, the algorithm of \cite{larsson3} may not work well here, as it may pick packets that have already been seen, which could cause non-innovative transmissions. 

The compatibility of our algorithm with random coding makes it particularly useful from an implementation perspective. With random coding, each receiver only needs to inform the sender the set of packets it has seen. There is no need to convey the exact knowledge space. This can be done simply by generating a TCP-like cumulative ACK upon seeing a packet. Thus, the ACK format is the same as in traditional ARQ-based schemes. Only its interpretation is different.  
\end{remark}

We next present the formal description and analysis of the queue update algorithm.

\subsubsection{The queuing module}\label{formal}
The algorithm works with the RREF bases of the receivers' knowledge spaces. The coefficient vectors are with respect to the current queue contents and not the original packet stream.

\noindent {\it \underline{Algorithm 2 (b)}}
\begin{enumerate}
\item [1.]Initialize matrices $B_1, B_2, \ldots, B_n$ to the empty matrix. 
These matrices will hold the bases of the incremental knowledge spaces of the receivers.  

\item [2.]{\it Incorporate new arrivals:}
 Suppose there are $a$ new arrivals. Add the new packets to the end of the queue. Append $a$ all-zero columns on the right to each $B_j$ for the new packets. 

\item [3.]{\it Transmission: }
If the queue is empty, do nothing; else compute $\mathbf{g}$ using the coding module and transmit it.

\item [4.]{\it Incorporate channel state feedback: }

For every receiver $j=1$ to $n$, do:

If receiver $j$ received the transmission, include the coefficient vector of $\mathbf{g}$ in terms of the current queue contents, as a new row in $B_j$. Perform Gaussian elimination.

\item [5.]{\it Separate out packets that all receivers have seen: } 

Update the following sets and bases:

\ \ $S_j'$\ \ := Set of packets corresponding to the pivot columns of $B_j$

\ \ $S_{\Delta}'$\ := $\cap_{j=1}^n S_j'$ 

New $B_j$\ := Sub-matrix of current $B_j$ obtained by excluding columns in $S_\Delta'$ and corresponding pivot rows.

\item [6.]{\it Update the queue: }
Drop the packets in $S_\Delta'$. 

\item [7.]Go back to step 2 for the next slot.
\end{enumerate}

\subsubsection{Connecting the physical and virtual queue sizes}\label{proofsection1}
The following theorem describes the asymptotic growth of the expected physical queue size under our new queuing rule. 

\ 

\begin{theorem}\label{mainthm}
\it For Algorithm 2 (b), the physical queue size at the sender is upper-bounded by the sum of the virtual queue sizes, \emph{i.e.}, the sum of the degrees-of-freedom backlog between the sender and the receivers. Hence, 
the expected size of the physical queue in steady state for Algorithm 2 (b) is $O\left(\frac1{1-\rho}\right)$.
\end{theorem}

\ 

In the rest of this section, we will prove the above result. Now, in order to relate the queue size to the backlog in number of degrees of freedom, we will need the following notation:

\ 

\noindent \ $S(t)$\ 	\ := Set of packets arrived at sender till the end of slot $t$

\noindent \ $V(t)$\ := Sender's knowledge space after incorporating the arrivals in slot $t$. This is simply equal to $\mathbb{F}_q^{|S(t)|}$

\noindent \ $V_j(t)$\ := Receiver $j$'s knowledge space at the end of slot $t$.  It is a subspace of $V(t)$.

\noindent \ $S_j(t)$\ := Set of packets receiver $j$ has seen till end of slot $t$

We will now formally argue that Algorithm 2 (b) indeed implements the drop-when-seen rule in spite of the incremental implementation. In any slot, the columns of $B_j$ are updated as follows. When new packets are appended to the queue, new columns are added to $B_j$ on the right. When packets are dropped from the queue, corresponding columns are dropped from $B_j$. There is no rearrangement of columns at any point. This implies that a one-to-one correspondence is always maintained between the columns of $B_j$ and the packets currently in the queue. Let $U_j(t)$ be the row space of $B_j$ at time $t$. Thus, if $(u_1, u_2, \ldots, u_{Q(t)})$ is any vector in $U_j(t)$, it corresponds to a linear combination of the form $\sum_{i=1}^{Q(t)} u_i \mathbf{p_i}$, where $\mathbf{p_i}$ is the $i^{th}$ packet in the queue at time $t$. The following theorem connects the incremental knowledge space $U_j(t)$ to the cumulative knowledge space $V_j(t)$.

\begin{theorem}
\it In Algorithm 2 (b), for each receiver $j$, at the end of slot $t$, for any $\mathbf{u}\in U_j(t)$, the linear combination $\sum_{i=1}^{Q(t)} u_i \mathbf{p_i}$ is known to the receiver $j$, where $\mathbf{p_i}$ denotes the $i^{th}$ packet in the queue at time $t$. 
\end{theorem}
\IEEEproof
We will use induction on $t$. For $t=0$, the system is completely empty and the statement is vacuously true. Let us now assume that the statement is true at time $(t-1)$. Consider the operations in slot $t$. A new row is added to $B_j$ only if the corresponding linear combination has been successfully received by receiver $j$. Hence, the statement is still true. Row operations involved in Gaussian elimination do not alter the row space. Finally, when some of the pivot columns are dropped along with the corresponding pivot rows in step 5, this does not affect the linear combinations to which the remaining rows correspond because the pivot columns have a 0 in all rows except the pivot row. Hence, the three operations that are performed between slot $(t-1)$ and slot $t$ do not affect the property that the vectors in the row space of $B_j$ correspond to linear combinations that are known at receiver $j$. This proves the theorem.
\endproof

If a packet corresponds to a pivot column in $B_j$, the corresponding pivot row is a linear combination of the packet in question with packets that arrived after it. From the above theorem, receiver $j$ knows this linear combination which means it has seen the packet. This leads to the following corollary.
\begin{corollary}
\it If a packet corresponds to a pivot column in $B_j$, then it has been seen by receiver $j$. 
\end{corollary}

Thus, in step 5, $S_\Delta'(t)$ consists of those packets in the queue that all receivers have seen by the end of slot $t$. In other words, the algorithm retains only those packets that have not yet been seen by all receivers. Even though the algorithm works with an incremental version of the knowledge spaces, namely $U_j(t)$, it maintains the queue in the same way as if it was working with the cumulative version $V_j(t)$. Thus, the incremental approach is equivalent to the cumulative approach.

We will require the following lemma to prove the main theorem.
\begin{lemma}\label{modularitylemma}
\it Let $A_1, A_2, \ldots, A_k$ be subsets of a set $A$. Then, for $k\ge 1$,
	\begin{equation}\label{connectingineq}
		|A|-|\cap_{i=1}^{k} A_i|\le \sum_{i=1}^k(|A|-|A_i|)
	\end{equation}
\end{lemma}

\IEEEproof
\begin{eqnarray*}
&&|A|-|\cap_{i=1}^{k} A_i|\\
&=&|A\cap(\cap_{i=1}^k A_i)^c| \mbox{\ \ (since the $A_i$'s are subsets of $A$)}\\
&=&|A\cap(\cup_{i=1}^k A_i^c)|\mbox{\ \ (by De Morgan's law)}\\
&=&|\cup_{i=1}^k (A\cap A_i^c)|\mbox{\ \ \ (distributivity)}\\
&\le &\sum_{i=1}^k |A\cap A_i^c|\mbox{\ \ (union bound)}\\
&=&\sum_{i=1}^k(|A|-|A_i|)
\end{eqnarray*}
\endproof

Now, we are ready to prove Theorem \ref{mainthm}.

{\it Proof of Theorem \ref{mainthm}:}
Since the only packets in the queue at any point are those that not all receivers have seen, we obtain the following expression for the physical queue size at the sender at the end of slot $t$:
\[Q(t)=|S(t)|-|\cap_{j=1}^n S_j(t)|\]

If we apply Lemma~\ref{modularitylemma} to the sets $S(t)$ and $S_j(t), j=1, 2, \ldots ,n$ then the left hand side of inequality (\ref{connectingineq}) becomes the sender queue size $Q(t)$ given above. Now, $|S_j(t)|=dim[V_j(t)]$, using Corollary \ref{numberseen}. Hence the right hand side of inequality (\ref{connectingineq}) can be rewritten as $\sum_{j=1}^n \big[ dim[V(t)] - dim[V_j(t)] \big]$, which is the sum of the virtual queue sizes. 

Finally, we can find the asymptotic behavior of the physical queue size in steady state under Algorithm 2 (b). Since the expected virtual queue sizes themselves are all $O\left(\frac1{1-\rho}\right)$ from Equation (\ref{vqsize}), we obtain the stated result.
\endproof

\subsubsection{The coding module}\label{codingmodule} 
We now present a coding module that is compatible with the drop-when-seen queuing algorithm in the sense that it always forms a linear combination using packets that are currently in the queue maintained by the queuing module. In addition, we show that the coding module satisfies the innovation guarantee property.

Let $\{u_1, u_2, \ldots , u_m\}$ be the set of indices of the next unseen packets of the receivers, sorted in ascending order (In general, $m\le n$, since the next unseen packet may be the same for some receivers). Exclude receivers whose next unseen packets have not yet arrived at the sender. Let $R(u_i)$ be the set of receivers whose next unseen packet is $\mathbf{p_{u_i}}$. We now present the coding module to select the linear combination for transmission.

\begin{enumerate}
	\item {\it Loop over next unseen packets}
		
		For $j=1$ to $m$, do:

			All receivers in $R(u_j)$ have seen packets $\mathbf{p_{u_i}}$ for $i<j$. Now, $\forall r \in R(u_j)$, find $\mathbf{y_r}:=\sum_{i=1}^{j-1} \alpha_i \mathbf{W_r}(\mathbf{p_{u_i}})$, where $\mathbf{W_r}(\mathbf{p_{u_i}})$ is the witness for receiver $r$ seeing $\mathbf{p_{u_i}}$. Pick $\alpha_j \in \mathbb{F}_q$ such that $\alpha_j$ is different from the coefficient of $\mathbf{p_{u_j}}$ in $\mathbf{y_r}$ for each $r\in R(u_j)$. 
	\item {\it Compute the transmit packet: } \ \ 
		$\mathbf{g}:= \sum_{i=1}^m \alpha_i \mathbf{p_{u_i}}$ 
\end{enumerate}

It is easily seen that this coding module is compatible with the drop-when-seen algorithm. Indeed, it does not use any packet that has been seen by all receivers in the linear combination. It only uses packets that at least one receiver has not yet seen. The queue update module retains precisely such packets in the queue. The next theorem presents a useful property of the coding module. 
\begin{theorem}\label{innovation}
	\it If the field size is at least $n$, then the coding module picks a linear combination that will cause any receiver to see its next unseen packet upon successful reception.
\end{theorem}

\IEEEproof
First we show that a suitable choice always exists for $\alpha_j$ that satisfies the requirement in step 1. For $r\in R(u_1)$, $\mathbf{y_r}=\mathbf{0}$. Hence, as long as $\alpha_1\neq 0$, the condition is satisfied. So, pick $\alpha_1=1$. Since at least one receiver is in $R(u_1)$, we have that for $j>1$, $|R(u_j)|\le (n-1)$. Even if each $\mathbf{y_r}$ for $r\in R(u_j)$ has a different coefficient for $\mathbf{p_{u_j}}$, that covers only $(n-1)$ different field elements. If $q\ge n$, then there is a choice left in $\mathbb{F}_q$ for $\alpha_j$. 

Now, we have to show that the condition given in step 1  implies that the receivers will be able to see their next unseen packet. Indeed, for all $j$ from 1 to $m$, and for all $r \in R(u_j)$, receiver $r$ knows $\mathbf{y_r}$, since it is a linear combination of witnesses of $r$. Hence, if $r$ successfully receives $\mathbf{g}$, it can compute $(\mathbf{g}-\mathbf{y_r})$. Now, $\mathbf{g}$ and $\mathbf{y_r}$ have the same coefficient for all packets with index less than $u_j$, and a different coefficient for $\mathbf{p_{u_j}}$. Hence, $(\mathbf{g}-\mathbf{y_r})$ will involve $\mathbf{p_{u_j}}$ and only packets with index beyond $u_j$. This means $r$ can see $\mathbf{p_{u_j}}$ and this completes the proof. 
\endproof

Theorem \ref{rreftheorem} implies that seeing an unseen packet corresponds to receiving an unknown degree of freedom. Thus, Theorem \ref{innovation} essentially says that the innovation guarantee property is satisfied and hence the scheme is throughput optimal.

This theorem is closely related to the result derived in \cite{larsson3} that computes the minimum field size needed to guarantee innovation. The difference is that our result uses the framework of seen packets to make a more general statement by specifying not only that innovation is guaranteed, but also that packets will be seen in order with this deterministic coding scheme. This means packets will be dropped in order at the sender.

\section{Overhead}
In this section, we comment on the overhead required for Algorithms 1 and 2 (b). There are several types of overhead.

\subsection{Amount of feedback}
Our scheme assumes that every receiver feeds back one bit after every slot, indicating whether an erasure occurred or not. In comparison, the drop-when-decoded scheme requires feedback only when packets get decoded. However, in that case, the feedback may be more than one bit -- the receiver will have to specify the list of all packets that were decoded, since packets may get decoded in groups. In a practical implementation of the drop-when-seen algorithm, TCP-like cumulative acknowledgments can be used to inform the sender which packets have been seen.

\subsection{Identifying the linear combination}
Besides transmitting a linear combination of packets, the sender must also embed information that allows the receiver to identify what linear combination has been sent. This involves specifying which packets have been involved in the combination, and what coefficients were used for these packets.

\subsubsection{Set of packets involved}
The baseline algorithm uses all packets in the queue for the linear combination. The queue is updated in a first-in-first-out (FIFO) manner, \ie, no packet departs before all earlier packets have departed. This is a consequence of the fact that the receiver signals successful decoding only when the virtual queue becomes empty\footnote{As mentioned earlier in Remark \ref{caveat}, we assume that the sender checks whether any packets have been newly decoded, only when the virtual queue becomes empty.}. The FIFO rule implies that specifying the current contents of the queue in terms of the original stream boils down to specifying the sequence number of the head-of-line packet and the last packet in the queue in every transmission. 

The drop-when-seen algorithm does not use all packets from the queue, but only at most $n$ packets from the queue (the next unseen packet of each receiver). This set can be specified by listing the sequence number of these $n$ packets. 

Now, in both cases, the sequence number of the original stream cannot be used as it is, since it grows unboundedly with time. However, we can avoid this problem using the fact that the queue contents are updated in a FIFO manner (This is also true of our drop-when-seen scheme -- the coding module guarantees that packets will be seen in order, thereby implying a FIFO rule for the sender's queue.). The solution is to express the sequence number relative to an origin that also advances with time, as follows. If the sender is certain that the receiver's estimate of the sender's queue starts at a particular point, then both the sender and receiver can reset their origin to that point, and then count from there. 

For the baseline case, the origin can be reset to the current HOL packet, whenever the receiver sends feedback indicating successful decoding. The idea is that if the receiver decoded in a particular slot, that means it had a successful reception in that slot. Therefore, the sender can be certain that the receiver must have received the latest update about the queue contents and is therefore in sync with the sender. Thus, the sender and receiver can reset their origin. Note that since the decoding epochs of different receivers may not be synchronized, the sender will have to maintain a different origin for each receiver and send a different sequence number to each receiver, relative to that receiver's origin. This can be done simply by concatenating the sequence number for each receiver in the header. 

To determine how many bits are needed to represent the sequence number, we need to find out what range of values it can take. In the baseline scheme, the sequence number range will be proportional to the busy period of the virtual queue, since this determines how often the origin is reset. Thus, the overhead in bits for each receiver will be proportional to the logarithm of the expected busy period, \ie, $O\left(\log_2\frac{1}{1-\rho}\right)$.

For the drop-when-seen scheme, the origin can be reset whenever the receiver sends feedback indicating successful reception. Thus, the origin advances a lot more frequently than in the baseline scheme.

\subsubsection{Coefficients used} 
The baseline algorithm uses a random linear coding scheme. Here, potentially all packets in the queue get combined in a linear combination. So, in the worst case, the sender would have to send one coefficient for every packet in the queue. If the queue has $m$ packets, this would require $m\log_2 q$ bits, where $q$ is the field size. In expectation, this would be $O\left(\frac{\log_2 q}{(1-\rho)^2}\right)$ bits. If the receiver knows the pseudorandom number generator used by the sender, then it would be sufficient for the sender to send the current state of the generator and the size of the queue. Using this, the receiver can generate the coefficients used by the sender in the coding process. The new drop-when-seen algorithm uses a coding module which combines the next unseen packet of each receiver. Thus, the overhead for the coefficients is at most $n\log_2 q$ bits, where $n$ is the number of receivers. It does not depend on the load factor $\rho$ at all. 
\subsection{Overhead at sender}
While Algorithm 2 (b) saves in buffer space, it requires the sender to store the basis matrix of each receiver, and update them in every slot based on feedback. However, storing a row of the basis matrix requires much less memory than storing a packet, especially for long packets. Thus, there is an overall saving in memory. The update of the basis matrix simply involves one step of the Gaussian elimination algorithm.

\subsection{Overhead at receiver}
The receiver will have to store the coded packets till they are decoded. It will also have to decode the packets. For this, the receiver can perform a Gaussian elimination after every successful reception. Thus, the computation for the matrix inversion associated with decoding can be spread over time. 

\section{Decoding delay}\label{delaysection}
With the coding module of Section \ref{codingmodule}, although a receiver can see the next unseen packet in every successful reception, this does not mean the packet will be decoded immediately. In general, the receiver will have to collect enough equations in the unknown packets before being able to decode them, resulting in a delay. We consider two notions of delay in this paper:

\begin{definition}[Decoding Delay]
  The \emph{decoding delay} of a packet with respect to a receiver is the time that elapses between the arrival of the packet at the sender and the decoding of the packet by the receiver under consideration.
\end{definition}

As discussed in Section \ref{intro}, some applications can make use of a packet only if all prior packets have been decoded. In other words, the application will accept packets only up to the front of contiguous knowledge. This motivates the following stronger notion of delay.

\begin{definition}[Delivery Delay]
  The \emph{delivery delay} of a packet with respect to a receiver is the time that elapses between the arrival of the packet at the sender and the delivery of the packet by the receiver to the application, with the constraint that packets may be delivered only in order.
\end{definition}

It follows from these definitions that \emph{the decoding delay is always less than or equal to the delivery delay}. Upon decoding the packets, the receiver will place them in a reordering buffer until they are delivered to the application.

In this section, we study the expectation of these delays for an arbitrary packet. It can be shown using ergodic theory that the long term average of the delay experienced by the packets in steady state converges to this expectation with high probability. We focus on the asymptotic growth of the expected delay as $\rho \rightarrow 1$.

The section is organized as follows. We first study the delivery delay behavior of Algorithms 1 and 2(b), and provide an upper bound on the asymptotic expected delivery delay for any policy that satisfies the innovation guarantee property. We then present a generic lower bound on the expected decoding delay. Finally, we present a new coding module for the case of three receivers which not only guarantees innovation, but also aims to minimize the delivery delay. We conjecture that this algorithm achieves a delivery delay whose asymptotic growth matches that of the lower bound. This behavior is verified through simulations.

\subsection{An upper bound on delivery delay}
We now present the upper bound on delay for policies that satisfy the innovation guarantee property. The arguments leading to this bound are presented below. 

\ 

\begin{theorem}\label{upperbound}
\it The expected delivery delay of a packet for any coding module that satisfies the innovation guarantee property is $O\left(\frac1{(1-\rho)^2}\right)$.
\end{theorem}

\ 

For any policy that satisfies the innovation guarantee property, the virtual queue size evolves according to the Markov chain in Figure \ref{markovchain}. The analysis of Algorithm 1 in Section \ref{algm1section} therefore applies to any coding algorithm that guarantees innovation. 

As explained in that section, the event of a virtual queue becoming empty translates to successful decoding at the corresponding receiver, since the number of equations now matches the number of unknowns involved. Thus, an arbitrary packet that arrives at the sender will get decoded by receiver $j$ at or before the next emptying of the $j^{th}$ virtual queue. In fact, it will get delivered to the application at or before the next emptying of the virtual queue. This is because, when the virtual queue is empty, every packet that arrived at the sender gets decoded. Thus, the front of contiguous knowledge advances to the last packet that the sender knows. 

The above discussion implies that Equation (\ref{deejay}) gives an upper bound on the expected delivery delay of an arbitrary packet. We thus obtain the result stated above.

We next study the decoding delay of Algorithm 2 (b). We define \emph{the decoding event} to be the event that all seen packets get decoded. Since packets are always seen in order, the decoding event guarantees that the front of contiguous knowledge will advance to the front of seen packets.

We use the term \emph{leader} to refer to the receiver which has seen the maximum number of packets at the given point in time. Note that there can be more than one leader at the same time. The following theorem characterizes sufficient conditions for the decoding event to occur.
\begin{theorem}\label{decodingevent}
\it The decoding event occurs in a slot at a particular receiver if in that slot:
\begin{enumerate}
\item [(a)] The receiver has a successful reception which results in an empty virtual queue at the sender; \ \ OR
\item [(b)] The receiver has a successful reception and the receiver was a leader at the beginning of the slot.
\end{enumerate}
\end{theorem}
\IEEEproof
Condition (a) implies that the receiver has seen all packets that have arrived at the sender up to that slot. Each packet at the sender is an unknown and each seen packet corresponds to a linearly independent equation. Thus, the receiver has received as many equations as the number of unknowns, and can decode all packets it has seen.

Suppose condition (b) holds. Let $\mathbf{p_k}$ be the next unseen packet of the receiver in question. The sender's transmitted linear combination will involve only the next unseen packets of all the receivers. Since the receiver was a leader at the beginning of the slot, the sender's transmission will not involve any packet beyond $\mathbf{p_k}$, since the next unseen packet of all other receivers is either $\mathbf{p_k}$ or some earlier packet. After subtracting the suitably scaled witnesses of already seen packets from such a linear combination, the leading receiver will end up with a linear combination that involves only $\mathbf{p_k}$. Thus the leader not only sees $\mathbf{p_k}$, but also decodes it. In fact, none of the sender's transmissions so far would have involved any packet beyond $\mathbf{p_k}$. Hence, once $\mathbf{p_k}$ has been decoded, $\mathbf{p_{k-1}}$ can also be decoded. This procedure can be extended to all unseen packets, and by induction, we can show that all unseen packets will be decoded.
\endproof

The upper bound proved in Theorem \ref{upperbound} is based on the emptying of the virtual queues. This corresponds only to case (a) in Theorem \ref{decodingevent}. The existence of case (b) shows that in general, the decoding delay will be strictly smaller than the upper bound. A natural question is whether this difference is large enough to cause a different asymptotic behavior, \ie, does Algorithm 2 (b) achieve a delay that asymptotically has a smaller exponent of growth than the upper bound as $\rho \rightarrow 1$? We conjecture that this is not the case, \ie, that the decoding delay for Algorithm 2 (b) is also $\Omega\left(\frac1{(1-\rho)^2}\right)$, although the constant of proportionality will be smaller. For the two receiver case, based on our simulations, this fact seems to be true. Figure \ref{delaysim} shows the growth of the decoding delay averaged over a large number of packets, as a function of $\frac{1}{(1-\rho)}$. The resulting curve seems to be close to the curve $\frac{0.37}{(1-\rho)^2}$, implying a quadratic growth. The value of $\rho$ ranges from 0.95 to 0.98, while $\mu$ is fixed to be 0.5. The figure also shows the upper bound based on busy period measurements. This curve agrees with the formula in Equation (\ref{deejay}) as expected.
\begin{figure}[t]
\centering
\includegraphics[width=0.45\textwidth]{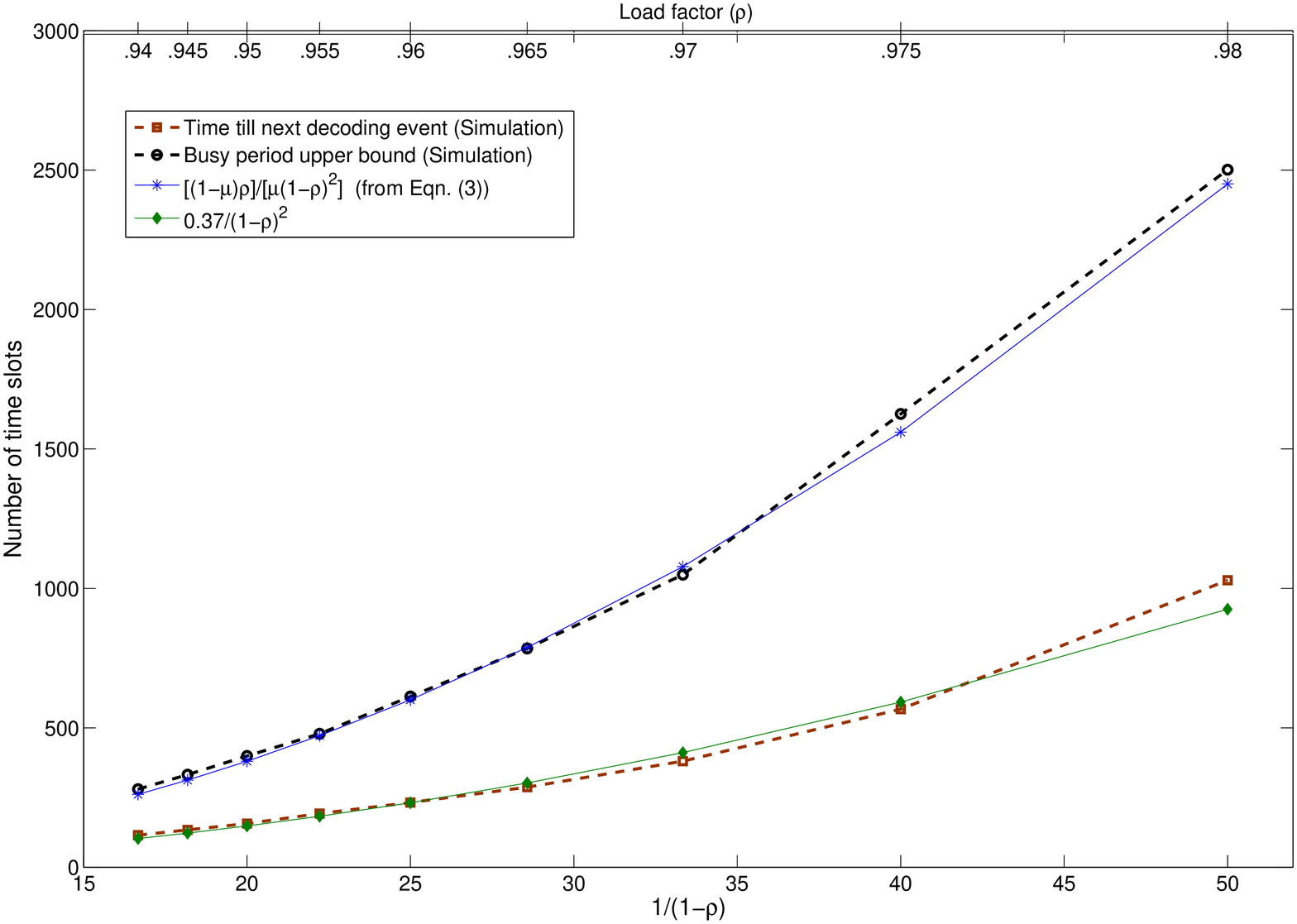}
\caption{Delay to decoding event and upper bound for 2 receiver case, as a function of $\frac1{(1-\rho)}$. The corresponding values of $\rho$ are shown on the top of the figure.}
\label{delaysim}
\end{figure}

\subsection{The lower bound}\label{lowerbound}
\begin{lemma}
  \it The expected per-packet delay is lower bounded by $\Omega\left(\frac1{1-\rho}\right)$
\end{lemma}
\IEEEproof
The expected per-packet delay for the single receiver case is clearly a lower bound for the corresponding quantity at one of the receivers in a multiple-receiver system. We will compute this lower bound in this section. Figure \ref{markovchain} shows the Markov chain for the queue size in the single receiver case. If $\rho=\frac\lambda\mu<1$, then the chain is positive recurrent and the steady state expected queue size can be computed to be $\frac{\rho(1-\mu)}{(1-\rho)}=\Theta\left(\frac1{1-\rho}\right)$ (see Equation (\ref{steadystatedist})). Now, if $\rho<1$, then the system is stable and Little's law can be applied to show that the expected per-packet delay in the single receiver system is also $\Theta\left(\frac1{1-\rho}\right)$. 
\endproof

\subsection{An alternate coding module for better delay}
In this section, we present a new coding module for the case of three receivers that significantly improves the delay performance compared to Algorithm 2 (b). In particular, we obtain 100\% throughput and conjecture that the algorithm simultaneously achieves asymptotically optimal decoding delay by meeting the lower bound of Lemma \ref{lowerbound}. The asymptotics here are in the realm of the load factor $\rho$ tending to 1 from below, while keeping either the arrival rate $\lambda$ or the channel quality parameter $\mu$ fixed at a number less than 1. 

We introduce a new notion of packets that a node has ``heard of''. 

\begin{definition}[Heard of a packet]
A node is said to have \emph{heard of} a packet if it knows some linear combination involving that packet.
\end{definition}

\subsection*{The new coding module}\label{newcodingmodule}
Our coding module works in the Galois field of size 3. At the beginning of every slot, the module has to decide what linear combination to transmit. Since there is full feedback, the module is fully aware of the current knowledge space of each of the three receivers. The coding algorithm is as follows: 

\begin{figure}
	\centering
		\includegraphics[width=0.45\textwidth]{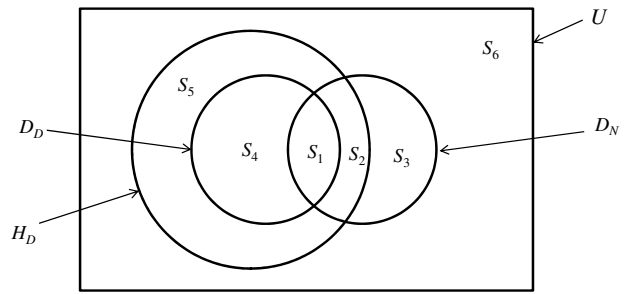}
		\caption{Sets used by the coding module}\label{venn}
\end{figure}
\begin{enumerate}
  \item Initialize $L=1, N=2, D=3, m=0$.
  \item Compute the following sets for all receivers $i=1, 2, 3$.

$H_i$:= Set of packets heard of by receiver $i$

$D_i$:= Set of packets decoded by receiver $i$

\item  Define a universe set $U$ consisting of packets $\mathbf{p}_1$ to $\mathbf{p}_m$, and also $\mathbf{p}_{m+1}$ if it has arrived. Compute the following sets\footnote{Notation: The subscripts $N$ and $D$ are simply indices. For example, $D_N$ is simply that $D_i$ for which $i=N$.} (See Figure \ref{venn}):
  \begin{itemize}
\item $S_1=D_N\cap D_D$
\item $S_2=D_N\cap(H_D\backslash D_D)$
\item $S_3=D_N\backslash H_D$
\item $S_4=D_D\backslash D_N$
\item $S_5=(H_D\backslash D_D)\backslash D_N$
\item $S_6=U\backslash(H_D\cup D_N)$
\end{itemize}

\item The coding module picks a linear combination depending on which of these sets $\mathbf{p}_{m+1}$ falls in, as follows:

  {\it Case 1 -- $\mathbf{p}_{m+1}$ has not arrived:} 
Check if both $S_2$ and $S_4$ are non-empty. If they are, pick the oldest packet from each, and send their sum. If not, try the pair of sets $S_3$ and $S_4$. If neither of these pairs of sets work, then send the oldest packet in $S_5$ if it is non-empty. If not, try $S_6$, $S_2$, $S_3$ and $S_4$ in that order. If all of these are empty, then send nothing.

{\it Case 2 -- $\mathbf{p}_{m+1}\in S_1$:} This is identical to case 1, except that $\mathbf{p}_{m+1}$ must also be added to the linear combination that case 1 suggests.

{\it Case 3 -- $\mathbf{p}_{m+1}\in S_2$:} Send $\mathbf{p}_{m+1}$ added to another packet. The other packet is chosen to be the oldest packet in the first non-empty set in the following list, tested in that order: $S_4$, $S_5$, $S_6$. (In the case where $\mathbf{p}_{m+1}\in S_2$, if the other packet $\mathbf{p}$ is chosen from $S_5$, then both the chosen packets are in $H_D\backslash D_D$. Therefore, the receiver $D$ might know one (but not both) of $(\mathbf{p}_{m+1}+\mathbf{p})$ or $(\mathbf{p}_{m+1}+2 \mathbf{p})$. Hence, the coefficient for $\mathbf{p}$ in the transmitted combination must be selected to be either 1 or 2, in such a way that the resulting linear combination is innovative to receiver $D$.)

{\it Case 4 -- $\mathbf{p}_{m+1}\in S_3$:} Send $\mathbf{p}_{m+1}$ added to another packet. The other packet is chosen to be the oldest packet in the first non-empty set in the following list, tested in that order: $S_4$, $S_5$, $S_6$.

{\it Case 5 -- $\mathbf{p}_{m+1}\in S_4$:} Send $\mathbf{p}_{m+1}$ added to another packet. The other packet is chosen to be the oldest packet in the first non-empty set in the following list, tested in that order: $S_2$, $S_3$, $S_6$.

{\it Case 6 -- All other cases:} Send $\mathbf{p}_{m+1}$ as it is.

\item Transmit the chosen linear combination and collect the feedback from all receivers. Using the feedback, update the sets $H_i$ and $D_i$ for all the receivers.
\item Set the new value of $m$ to be the maximum of the ranks of the three receivers. Identify the set of receivers that have decoded all packets from 1 to $m$. If there is no such receiver, assign `L', `N' and `D' arbitrarily and go to step 3. (We show in Theorem \ref{leadertheorem} that there will always be at least one such receiver.) 
  
If there is more than one such receiver, pick the one with the lowest index to be `L'. Compute the \emph{unsolved set} $T_i:=H_i\backslash D_i$ for the other two receivers. If exactly one of them has a non-empty unsolved set, pick that receiver to be `D' (for deficit), and the other one to be `N' (for no deficit\footnote{If $H_i\backslash D_i$ is not empty, this indicates a deficit of equations compared to the unknowns involved in them.}). If neither has an unsolved set or if both have an unsolved set, assign `D' and `N' arbitrarily. (We show in Theorem \ref{deficittheorem} that at most one of them will have a non-empty unsolved set.) Go to step 3.
\end{enumerate}

\subsection{Properties of the coding module}
The above algorithm aims to guarantee innovation using as little mixing of packets as possible. In this section, we state and prove some key properties of the coding module, including the innovation guarantee property. In what follows, we use the notation $m(t)$ to denote the maximum rank among the three receivers at the beginning of slot $t$.

\begin{lemma}\label{limitlemma}
\it For any $t>0$, the transmission in any slot from $1$ to $t$ does not involve a packet with index beyond $m(t)+1$.
\end{lemma}
\IEEEproof
The proof is by induction on the slot number. 

\noindent {\it Basis step:} If anything is sent in slot 1, it has to be $\mathbf{p}_1$, since all the sets except $S_6$ are empty. Thus, as $m(1)=0$, the statement holds.

\noindent {\it Induction hypothesis:} Suppose no transmission up to and including slot $t$ has involved packets beyond $\mathbf{p}_{m(t)+1}$.

\noindent {\it Induction step:} Then at the beginning of slot $(t+1)$, the sets $S_1$ to $S_5$ cannot contain packets beyond $\mathbf{p}_{m(t)+1}$. Along with the definition of $S_6$ and the fact that $m(t+1)\ge m(t)$, this statement implies that $S_1$ to $S_6$ cannot contain any packet with index beyond $m(t+1)+1$. 

The coding module combines $\mathbf{p}_{m(t+1)+1}$ with up to 2 other packets from these sets. Thus, the resulting transmission will not involve any packet with index beyond $m(t+1)+1$.
\endproof

\begin{theorem}\label{leadertheorem}
  \it At the beginning of any slot $t>0$, at least one receiver has decoded all packets from $\mathbf{p}_1$ to $\mathbf{p}_{m(t)}$.
\end{theorem}
\IEEEproof
The proof is by induction on the slot number. 

\noindent {\it Basis step:} Since $m(1)=0$, the statement is trivially true for $t=1$.

\noindent {\it Induction hypothesis:} Suppose at the beginning of slot $t$, there is a receiver $R^*$ that has decoded all packets from $\mathbf{p}_1$ to $\mathbf{p}_{m(t)}$. 

\noindent {\it Induction step:} We need to show that the statement holds at the beginning of slot $(t+1)$. Clearly, $m(t)\le m(t+1)\le m(t)+1$ (The rank cannot jump by more than 1 per slot).

If $m(t+1)=m(t)$, then the statement clearly holds, as $R^*$ has already decoded packets from $\mathbf{p}_1$ to $\mathbf{p}_{m(t)}$. If $m(t+1)=m(t)+1$, then let $R'$ be the receiver with that rank. From Lemma \ref{limitlemma}, all transmissions up to and including the one in slot $t$, have involved packets with index 1 to $m(t)+1$. This means $R'$ has $m(t+1)$ linearly independent equations in the unknowns $\mathbf{p}_1$ to $\mathbf{p}_{m(t+1)}$. Thus, $R'$ can decode these packets and this completes the proof.
\endproof

\begin{definition}[Leader]
  In the context of this coding module, the node that has decoded all packets from $\mathbf{p}_1$ to $\mathbf{p}_{m(t)}$ at the beginning of slot $t$ is called the \emph{leader}. If there is more than one such node, then any one of them may be picked. 
\end{definition}

Note that the node labeled `L' in the algorithm corresponds to the leader. The other two nodes are called \emph{non-leaders}. We now present another useful feature of the coding module.

\begin{lemma}\label{minimalmixing}
\it From any receiver's perspective, the transmitted linear combination involves at most two undecoded packets in any slot.
\end{lemma}
\IEEEproof
The module mixes at most two packets with each other, except in case 2 where sometimes three packets are mixed. Even in case 2, one of the packets, namely $\mathbf{p}_{m+1}$, has already been decoded by both non-leaders, as it is in $S_1$. From the leader's perspective, there is only one unknown packet that could be involved in any transmission, namely, $\mathbf{p}_{m+1}$ (from Lemma \ref{limitlemma}). Thus, in all cases, no more than two undecoded packets are mixed from any receiver's point of view.
\endproof

\subsubsection*{Structure of the knowledge space}
The above property leads to a nice structure for the knowledge space of the receivers. In order to explain this structure, we define the following relation with respect to a specific receiver. The ground set $G$ of the relation contains all packets that have arrived at the sender so far, along with a fictitious all-zero packet that is known to all receivers even before transmission begins. Note that the relation is defined with respect to a specific receiver. Two packets $\mathbf{p_x}\in G$ and $\mathbf{p_y}\in G$ are defined to be related to each other if the receiver knows at least one of $\mathbf{p_x}+\mathbf{p_y}$ and $\mathbf{p_x}+2\mathbf{p_y}$.  

\begin{lemma}\label{equivrelation}
\it The relation defined above is an equivalence relation. 
\end{lemma}
\IEEEproof
A packet added with two times the same packet gives $\mathbf{0}$ which is trivially known to the receiver. Hence, the relation is reflexive. The relation is symmetric because addition is a commutative operation. For any $\mathbf{p_x}, \mathbf{p_y}, \mathbf{p_z}$ in $G$, if a receiver knows $\mathbf{p_x}+ \alpha\mathbf{p_y}$ and $\mathbf{p_y}+ \beta\mathbf{p_z}$, then it can compute either  $\mathbf{p_x}+\mathbf{p_z}$ or $\mathbf{p_x}+2\mathbf{p_z}$ by canceling out the $\mathbf{p_y}$, for $\alpha=1$ or 2 and $\beta=1$ or 2. Therefore the relation is also transitive and is thus an equivalence relation. 
\endproof

The relation defines a partition on the ground set, namely the equivalence classes, which provide a structured abstraction for the knowledge of the node. The reason we include a fictitious all-zero packet in the ground set is that it allows us to represent the decoded packets within the same framework. It can be seen that the class containing the all-zero packet is precisely the set of decoded packets. Packets that have not been involved in any of the successfully received linear combinations so far will form singleton equivalence classes. These correspond to the packets that the receiver has not heard of. All other classes contain the packets that have been heard of but not decoded. Packets in the same class are equivalent in the sense that revealing any one of them will reveal the entire class to the receiver. 

\begin{theorem}\label{deficittheorem}
\it At the beginning of any slot $t>0$, at least one of the two non-leaders has an empty unsolved set, \ie, has $H_i= D_i$.
\end{theorem}
\IEEEproof
Initially, every receiver has an empty unsolved set ($H_i\backslash D_i$). It becomes non-empty only when a receiver receives a mixture involving two undecoded packets. It can be verified that this happens only in two situations: 
\begin{enumerate}
  \item When case 4 occurs, and $\mathbf{p}_{m+1}\in S_3$ is mixed with a packet from $S_6$; or 
  \item When case 5 occurs, and $\mathbf{p}_{m+1}\in S_4$ is mixed with a packet from $S_6$. 
\end{enumerate}
Even in these cases, only one receiver develops an unsolved set because, from the other two receivers' perspective, the mixture involves one decoded packet and one new packet. 

The receiver that develops an unsolved set, say node $j$, is labeled `D' in step 6, and $H_D\backslash D_D$ now contains two packets. Let the slot in which this happens for the first time be $t_1$. Now, at least one of these two packets is in $S_2$ because, as argued above, each of the other two receivers has decoded one of these packets. So, no matter which of the other two receivers is labeled `N', one of these two packets has already been decoded by `N'. 

We will now prove by contradiction that neither of the other two nodes can develop an unsolved set, as long as node $j$'s unsolved set is not empty. In other words, node $j$ will continue to be labeled as `D', until its unsolved set is fully decoded.

Suppose one of the other nodes, say node $i$ ($i\ne j$), indeed develops an unsolved set while $H_D\backslash D_D$ is still non-empty. Let $t_2$ be the slot when this happens. Thus, from slot $t_1+1$ to slot $t_2$, node $j$ is labeled $D$. We track the possible changes to $H_D\backslash D_D$ in terms of its constituent equivalence classes, during this time. Only three possible types of changes could happen:
\begin{enumerate}
  \item {\it Addition of new class:} A new equivalence class will be added to $H_D\backslash D_D$ if case 4 occurs, and $\mathbf{p}_{m+1}\in S_3$ is mixed with a packet from $S_6$. In this case, the new class will again start with two packets just as above, and at least one of them will be in $S_2$.
  \item {\it Decoding of existing class:} An existing equivalence class could get absorbed into the class of decoded packets if an innovative linear combination is revealed about the packets in the class, allowing them to be decoded.
  \item {\it Expansion of existing class:} If a linear combination involves a packet in an existing class and a new unheard of packet, then the new packet will simply join the class. 
\end{enumerate}

In every class, at least one of the initial two packets is in $S_2$ when it is formed. The main observation is that during the period up to $t_2$, this remains true till the class gets decoded. The reason is as follows. Up to slot $t_2$, node $j$ is still called `D'. Even if the labels `L' and `N' get interchanged, at least one of the initial pair of packets will still be in $D_N$, and therefore in $S_2$. The only way the class's contribution to $S_2$ can become empty is if the class itself gets decoded by $D$. 

This means, as long as there is at least one class, \ie, as long as $H_D\backslash D_D$ is non-empty, $S_2$ will also be non-empty. In particular, $S_2$ will be non-empty at the start of slot $t_2$. 

By assumption, node $i$ developed an unsolved set in slot $t_2$. Then, node $i$ could not have been a leader at the beginning of slot $t_2$ -- a leader can never develop an unsolved set, as there is only one undecoded packet that could ever be involved in the transmitted linear combination, namely $\mathbf{p}_{m+1}$ (Lemma \ref{limitlemma}). Therefore, for node $i$ to develop an unsolved set, it has to first be a non-leader, \ie, `N' at the start of slot $t_2$. In addition, case 5 must occur, and $\mathbf{p}_{m+1}\in S_4$ must get mixed with a packet from $S_6$ during $t_2$. But this could not have happened, as we just showed that $S_2$ is non-empty. Hence, in case 5, the coding module would have preferred $S_2$ to $S_6$, thus leading to a contradiction. 

Once $j$'s unsolved set is solved, the system returns to the initial state of all unsolved sets being empty. The same argument applies again, and this proves that a node cannot develop an unsolved set while another already has a non-empty unsolved set.
\endproof

\subsubsection*{Innovation guarantee}
Next, we prove that the coding module provides the innovation guarantee. 

\ 

\begin{theorem}\label{innovationtheorem}
\it The transmit linear combination computed by the coding module is innovative to all receivers that have not decoded everything that the sender knows.
\end{theorem}

\ 

\IEEEproof
Since the maximum rank is $m$, any deficit between the sender and any receiver will show up within the first $(m+1)$ packets. Thus, it is sufficient to check whether $U\backslash D_i$ is non-empty, while deciding whether there is a deficit between the sender and receiver $i$.

Consider the leader node. It has decoded packets $\mathbf{p}_1$ to $\mathbf{p}_m$ (by Theorem \ref{leadertheorem}). If $\mathbf{p}_{m+1}$ has not yet arrived at the sender, then the guarantee is vacuously true. If $\mathbf{p}_{m+1}$ has arrived, then the transmission involves this packet in all the cases, possibly combined with one or two packets from $\mathbf{p}_1$ to $\mathbf{p_m}$, all of which the leader has already decoded. Hence, the transmission will reveal $\mathbf{p}_{m+1}$, and in particular, will be innovative.

Next, consider node `N'. If there is a packet in $U\backslash D_N$, then at least one of $S_4, S_5$ and $S_6$ will be non-empty. Let us consider the coding module case by case. 

\noindent {\it Case 1 -- } Suppose $S_4$ is empty, then the module considers $S_5$ and $S_6$ before anything else, thereby ensuring innovation. Suppose $S_4$ is not empty, then a packet from $S_4$ is mixed with a packet from $S_2$ or $S_3$ if available. Since $S_2$ and $S_3$ have already been decoded by `N', this will reveal the packet from $S_4$. If both $S_2$ and $S_3$ are empty, then $S_5, S_6$ and $S_4$ are considered in that order. Therefore, in all cases, if there is a deficit, an innovative packet will be picked. 

\noindent {\it Case 2 -- } This is identical to case 1, since $\mathbf{p}_{m+1}$ has already been decoded by `N'.

\noindent {\it Case 3 and 4 -- } $\mathbf{p}_{m+1}$ has already been decoded by `N', and the other packet is picked from $S_4, S_5$ or $S_6$, thus ensuring innovation.

\noindent {\it Case 5 and 6 -- } In these cases, $\mathbf{p}_{m+1}$ has not yet been decoded by `N', and is involved in the transmission. Since `N' has no unsolved set (Theorem \ref{deficittheorem}), innovation is ensured.

Finally, consider node `D'. If there is a packet in $U\backslash D_D$, then at least one of $S_2, S_3, S_5$ and $S_6$ will be non-empty. Again, we consider the coding module case by case.

\noindent {\it Case 1 -- } If $S_4$ is empty, the coding module considers $S_5, S_6, S_2$ or $S_3$ and reveals a packet from the first non-empty set. If $S_4$ is not empty, then then a packet from $S_4$ is mixed with a packet from $S_2$ or $S_3$ if available. Since $S_4$ has already been decoded by `D', this will reveal a packet from $S_2$ or $S_3$ respectively. If both $S_2$ and $S_3$ are empty, then $S_5$ and $S_6$ are considered. Thus, innovation is ensured.

\noindent {\it Case 2 -- } This is identical to case 1, since $\mathbf{p}_{m+1}$ has already been decoded by `D'.

\noindent {\it Case 3 -- } In this case, $\mathbf{p}_{m+1}\in H_D\backslash D_D$. There are four possibilities:

\begin{enumerate}
  \item If it is mixed with a packet from $S_4$, then since $D$ has already all packets in $S_4$, it will decode $\mathbf{p}_{m+1}$. 
  \item If instead it is mixed with a packet, say $\mathbf{p}$ from $S_5$, then since both packets have been heard of, it is possible that `D' already knows at most one of $\mathbf{p}+\mathbf{p}_{m+1}$ and $2 \mathbf{p}+\mathbf{p}_{m+1}$. Then, as outlined in step 4 of the algorithm (case 3), the coefficient of $\mathbf{p}$ is chosen so as to guarantee innovation. 
  \item If it is mixed with a packet from $S_6$, then innovation is ensured because the packet in $S_6$ has not even been heard of.
  \item If it is not mixed with any other packet, then also innovation is ensured, since $\mathbf{p}_{m+1}$ has not yet been decoded.
\end{enumerate}

\noindent {\it Case 4 -- } The exact same reasoning as in Case 3 holds here, except that the complication of picking the correct coefficient in possibility number 2 above, does not arise.

\noindent {\it Case 5 -- } In this case, $\mathbf{p}_{m+1}$ has already been decoded. The module considers $S_2, S_3$ and $S_6$. There is no need to consider $S_5$ because, if $S_5$ is non-empty, then so is $S_2$. This fact follows from the arguments in the proof of Theorem \ref{deficittheorem}.

\noindent {\it Case 6 -- } In all the other cases, $\mathbf{p}_{m+1}$ has not been decoded, and will therefore be innovative.
\endproof 

\subsection{Delay performance of the new coding module}
We now study the delay experienced by an arbitrary arrival before it gets decoded by one of the receivers. We consider a system where $\mu$ is fixed at 0.5. The value of $\rho$ is varied from 0.9 to 0.99 in steps of 0.01. We plot the expected decoding delay and delivery delay per packet, averaged across the three receivers, as a function of $\left(\frac{1}{1-\rho}\right)$ in Figure \ref{linearplot}. We also plot the log of the same quantities in Figure \ref{loglogplot}. The value of the delay is averaged over $10^6$ time slots for the first five points and $2\times 10^6$ time slots for the next three points and $5\times 10^6$ for the last two points. 

Figure \ref{linearplot} shows that the growth of the average decoding delay as well as the average delivery delay are linear in $\left(\frac{1}{1-\rho}\right)$ as $\rho$ approaches 1. Figure \ref{loglogplot} confirms this behavior -- we can see that the slopes on the plot of the logarithm of these quantities is indeed close to 1. This observation leads to the following conjecture:

\begin{figure}
	\centering
		\includegraphics[width=0.45\textwidth]{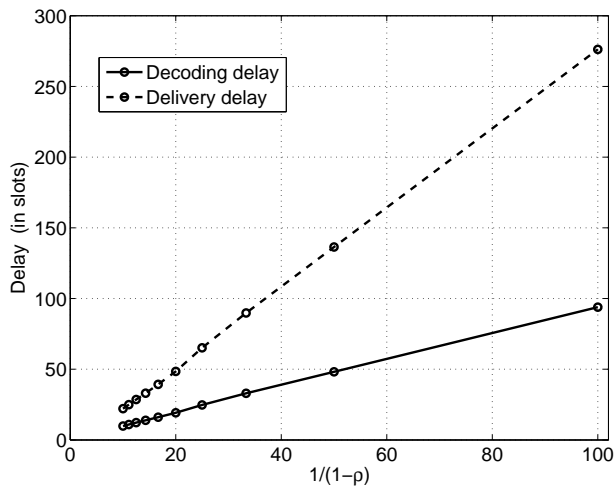}
		\caption{Decoding and delivery delay for the coding module in Section \ref{newcodingmodule}}\label{linearplot}
\end{figure}

\begin{figure}
	\centering
		\includegraphics[width=0.45\textwidth]{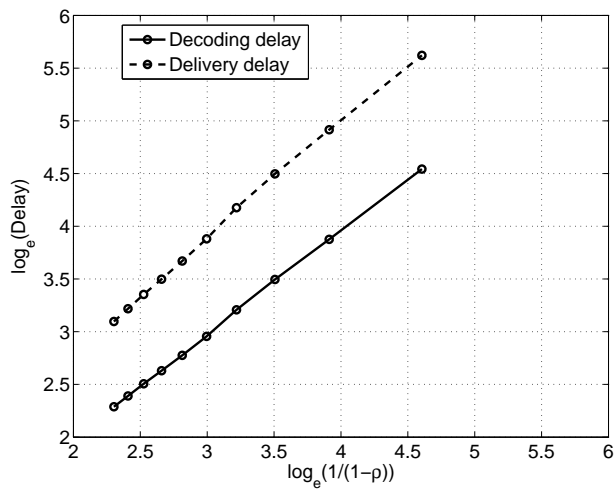}
		\caption{Log plot of the delay for the coding module in Section \ref{newcodingmodule}}\label{loglogplot}
\end{figure}

\begin{conjecture}\label{conj1}
\it For the newly proposed coding module, the expected decoding delay per packet, as well as the expected delivery delay per packet from a particular receiver's point of view grow as $O\left(\frac{1}{1-\rho}\right)$, which is asymptotically optimal.
\end{conjecture}

This conjecture, if true, implies that such feedback-based coding for delay also simplifies the queue management at the sender. If the sender simply follows a drop-when-decoded strategy, then by Little's theorem, the expected queue size of undecoded packets will be proportional to the expected decoding delay $O\left(\frac1{1-\rho}\right)$, which is asymptotically optimal. 

\section{Applications and further extensions}\label{extensions}
Although we have presented the algorithm in the context of a single packet erasure broadcast channel, we believe the main ideas in the scheme are quite robust and can be applied to more general topologies. The scheme readily extends to a tandem network of broadcast links (with no mergers) if the intermediate nodes use the witness packets in place of the original packets. We expect that it will also extend to other topologies with suitable modifications. In addition, we believe the proposed scheme will also be robust to delayed or imperfect feedback, just like conventional ARQ. Such a generalization can lead to a TCP-like protocol for systems that use network coding \cite{TCPNC}.

We have assumed the erasures to be independent and identically distributed across receivers. However, the analysis for Algorithm 2 (b) will hold even if we allow adversarial erasures. This is because, the guarantee that the physical queue size tracks the backlog in degrees of freedom is not a probabilistic guarantee, but a combinatorial guarantee on the instantaneous value of the queue sizes. Note that, while the erasures can be chosen adversarially, we will require the adversary to guarantee a certain minimum long-term connection rate from the sender to every receiver, so that the virtual queues can themselves be stabilized.

From a theoretical point of view, our results mean that any stability results or queue size bounds in terms of virtual queues can be translated to corresponding results for the physical queues. In addition, results from traditional queuing theory about M/G/1 queues or a Jackson network type of result \cite{desmondjournal} can be extended to the physical queue size in coded networks, as opposed to just the backlog in degrees of freedom. From a practical point of view, if the memory at the sender has to be shared among several different flows, then this reduction in queue occupancy will prove quite useful in getting statistical multiplexing benefits. 

For instance, one specific scenario where our results can be immediately applied is the multicast switch with intra-flow network coding, studied in \cite{infocom07}. The multicast switch has broadcast-mode links from each input to all the outputs. ``Erasures'' occur because the scheduler may require that only some outputs can receive the transmission, as the others are scheduled to receive a different transmission from some other input. In this case, there is no need for explicit feedback, since the sender can track the states of knowledge of the receivers simply using the scheduling configurations from the past. The results stated in \cite{infocom07} in terms of the virtual queues can thus be extended to the physical queues as well. 

Another important extension that needs to be investigated in the future, is the extension of the coding scheme for optimizing decoding and delivery delay to the case of more than three receivers. This problem is particularly important for real-time data streaming applications. 

\section{Conclusions}\label{conc}
In this work, we have presented a completely online approach to network coding based on feedback, which does not compromise on throughput and yet, provides benefits in terms of queue occupancy at the sender and decoding delay at the receivers. 

The notion of seen packets introduced in this work, allows the application of tools and results from traditional queuing theory in contexts that involve coding across packets. Using this notion, we proposed the drop-when-seen algorithm, which allows the physical queue size to track the backlog in degrees of freedom, thereby reducing the amount of storage used at the sender. Comparing the results in Theorem \ref{algm1qsize} and Theorem \ref{mainthm}, we see that the newly proposed Algorithm 2 (b) gives a significant improvement in the expected queue size at the sender, compared to Algorithm 1. 

For the three receiver case, we have proposed a new coding scheme that makes use of feedback to dynamically adapt the code in order to ensure low decoding delay. As argued earlier, $\Theta\left(\frac1{1-\rho}\right)$ is an asymptotic lower bound on the decoding delay and the stronger notion of delivery delay in the limit of the load factor approaching capacity ($\rho\rightarrow 1$). We conjecture that our scheme achieves this lower bound. If true, this implies the asymptotic optimality of our coding module in terms of both decoding delay and delivery delay. We have verified this conjecture through simulations. 

In summary, we believe that the proper combination of feedback and coding in erasure networks presents a wide range of benefits in terms of throughput, queue management and delay. Our work is a step towards realizing these benefits.
\bibliographystyle{IEEEtran}
\bibliography{References}

\appendices
\section{Derivation of the first passage time}\label{derivation}
	Consider the Markov chain $\{Q_j(t)\}$ for the virtual queue size, shown in Figure \ref{markovchain}. Assume that the Markov chain has an initial distribution equal to the steady state distribution (Equivalently, assume that the Markov chain has reached steady state.). We use the same notation as in Section \ref{algm1section}. 
	
	Define $N_m:=\inf\{t\ge 1: Q_j(t)=m\}$. We are interested in deriving for $k\ge 1$, an expression for $\Gamma_{k,0}$, the expected first passage time from state $k$ to 0, \ie, 
	\[\Gamma_{k,0} = \mathbb{E}[N_0|Q_j(0)=k]\]

Define for $i\ge 1$:
	\[X_i:=a(i)-d(i)\]
 where, $a(i)$ is the indicator function for an arrival in slot $i$, and $d(i)$ is the indicator function for the channel being on in slot $i$. Let $S_t:=\sum_{i=1}^t X_i$. If $Q_j(t)>0$, then the channel being on in slot $t$ implies that there is a departure in that slot. Thus the correspondence between the channel being on and a departure holds for all $0\le t \le N_0$. This implies that:
 \[\mbox{ For } t\le N_0, Q_j(t)=Q_j(0)+S_t\]
 Thus, $N_0$ can be redefined as the smallest $t\ge 1$ such that $S_t$ reaches $-Q_j(0)$. Thus, $N_0$ is a valid stopping rule for the $X_i$'s which are themselves IID, and have a mean $\mathbb{E}[X]=(\lambda-\mu)$. We can find $\mathbb{E}[N_0]$ using Wald's equality:
 \[\mathbb{E}[S_{N_0}|Q_j(0)=k]=\mathbb{E}[N_0|Q_j(0)=k]\cdot \mathbb{E}[X]\]
 \[\mbox{ i.e., } -k=\mathbb{E}[N_0|Q_j(0)=k]\cdot(\lambda-\mu)\]
 which gives:
 \[\Gamma_{k,0}=\mathbb{E}[N_0|Q_j(0)=k] = \frac{k}{\mu-\lambda}\]

\section{Proof of Lemma \ref{distrib_lemma}}\label{distriblemmaproof}
\IEEEproof
For any $z\in V_\Delta \oplus \cap_{i=1}^n U_i$, there is a $x\in V_\Delta$ and $y\in \cap_{i=1}^n U_i$ such that $z=x+y$. Now, for each $i$, $y\in U_i$. Thus, $z=x+y$ implies that $z\in \cap_{i=1}^n [V_\Delta \oplus U_i]$. Therefore, 
$V_\Delta \oplus \cap_{i=1}^n U_i \subseteq \cap_{i=1}^n [V_\Delta \oplus U_i]$.

Now, let $w\in \cap_{i=1}^n V_\Delta \oplus U_i$. Then for each $i$, there is a $x_i\in V_\Delta$ and $y_i\in U_i$ such that $w=x_i+y_i$. But, $w=x_i+y_i=x_j+y_j$ means that $x_i-x_j=y_i-y_j$. Now, $(x_i-x_j)\in V_\Delta$ and $(y_i-y_j)\in (U_1+U_2+\ldots +U_n)$. By hypothesis, these two vector spaces have only 0 in common. Thus, $x_i-x_j=y_i-y_j=0$. All the $x_i$'s are equal to a common $x\in V_\Delta$ and all the $y_i$'s are equal to a common $y$ which belongs to all the $U_i$'s. This means, $w$ can be written as the sum of a vector in $V_\Delta$ and a vector in $\cap_{i=1}^n U_i$, thereby proving that $\cap_{i=1}^n [V_\Delta \oplus U_i] \subseteq V_\Delta \oplus \cap_{i=1}^n U_i$.
\endproof

\section{Proof of Lemma \ref{associativity}}\label{associativityproof}
\IEEEproof
 Statement \ref{statement1} follows from the fact that $B$ is a subset of $B\oplus C$. Hence, if $A\cap (B\oplus C)$ is empty, so is $A\cap B$. 

 For statement \ref{statement2}, we need to show that $(A\oplus B)\cap C=\{\mathbf{0}\}$. Consider any element $\mathbf{x}\in (A\oplus B)\cap C$. Since it is in $A\oplus B$, there exist unique $\mathbf{a}\in A$ and $\mathbf{b}\in B$ such that $\mathbf{x}=\mathbf{a}+\mathbf{b}$. Now, since $\mathbf{b}\in B$ and $\mathbf{x}\in C$, it follows that $\mathbf{a}=\mathbf{x}-\mathbf{c}$ is in $B\oplus C$. It is also in $A$. Since $A$ is independent of $B\oplus C$, $\mathbf{a}$ must be $\mathbf{0}$. Hence, $\mathbf{x}=\mathbf{b}$. But this means $\mathbf{x}\in B$. Since it is also in $C$, it must be $\mathbf{0}$, as $B$ and $C$ are independent. This shows that the only element in $(A\oplus B)\oplus C$ is $\mathbf{0}$. 

Statement \ref{statement3} can be proved as follows. 

$\mathbf{x}\in A\oplus(B\oplus C)$
\begin{align*}
  \Leftrightarrow&\exists \mbox{ unique } \mathbf{a}\in A, \mathbf{d}\in B\oplus C \mbox{ s.t. } \mathbf{x}=\mathbf{a}+\mathbf{d}\\
  \Leftrightarrow&\exists \mbox{ unique } \mathbf{a}\in A, \mathbf{b}\in B, \mathbf{c}\in C \mbox{ s.t. } \mathbf{x}=\mathbf{a}+\mathbf{b}+\mathbf{c}\\
  \Leftrightarrow&\exists \mbox{ unique } \mathbf{e}\in A\oplus B, \mathbf{c}\in C \mbox{ s.t. } \mathbf{x}=\mathbf{e}+\mathbf{c}\\
 \Leftrightarrow&\mathbf{x}\in (A\oplus B)\oplus C
\end{align*}
\endproof

\end{document}